\documentclass[preprint,11pt,a4paper]{elsarticle}
\usepackage{amsfonts}
\usepackage{amsmath}
\usepackage{amssymb}
\usepackage{graphicx}
\usepackage{version}
\usepackage{geometry}

\journal{JMVA}
\newtheorem{theorem}{Theorem}
\newtheorem{acknowledgement}[theorem]{Acknowledgement}

\newtheorem{definition}[theorem]{Definition}
\newtheorem{example}[theorem]{Example}

\newtheorem{lemma}[theorem]{Lemma}

\newtheorem{remark}[theorem]{Remark}

\newenvironment{proof}[1][Proof]{\noindent\textbf{#1.} }{\ \rule{0.5em}{0.5em}}
\input{tcilatex}
\begin{document}

\begin{frontmatter}%


%

\title{Harmonic analysis and distribution-free inference for spherical distributions}%

%

\author{S. Rao Jammalamadaka}%

\address{Department of Statistics and Applied Probability, University of California, Santa Barbara, CA 93106, USA \\ rao@pstat.ucsb.edu}%

\author{Gyorgy Terdik \small{(Corresponding author)}}%

\address{Faculty of Informatics, University of Debrecen, 4029 Debrecen, Hungary \\ Terdik.Gyorgy@inf.unideb.hu}%

\begin{abstract}%

Fourier analysis and representation of \textit{circular distributions} in
terms of their Fourier coefficients, is quite commonly discussed and used
for model-free inference such as testing uniformity and symmetry etc. in
dealing with 2-dimensional directions. However a similar discussion for 
\textit{spherical distributions}, which are used to model 3-dimensional
directional data, has not been fully developed in the literature in terms of
their harmonics. This paper, in what we believe is the first such attempt,
looks at the probability distributions on a unit sphere, through the
perspective of spherical harmonics, analogous to the Fourier analysis for
distributions on a unit circle. Harmonic representations of many currently
used spherical models are presented and discussed. A very general family of
spherical distributions is then introduced, special cases of which yield
many known spherical models. Through the prism of harmonic analysis, one can
look at the mean direction, dispersion, and various forms of symmetry for
these models in a generic setting. Aspects of distribution free inference
such as estimation and large-sample tests for these symmetries, are
provided. The paper concludes with a real-data example analyzing the
longitudinal sunspot activity.

\end{abstract}%


\end{frontmatter}%



\section{Introduction}

Probability models for directional data in two- and three-dimensions can be
represented using the circumference of a unit circle or the surface of a
unit sphere as the support, and are called circular and spherical
distributions respectively. In any discussion of probability distributions
on the circle, Fourier analysis becomes an integral part, because such
probability densities on a circle are periodic with period $2\pi$ (see e.g. 
\cite{jammalamadaka2001topics}, p. 25, \cite{mardia2009directional}). On the
other hand, when it comes to directions in 3-dimensions and probability
distributions on the sphere, such analysis is nearly absent, partly because
of the complexity, as we shall soon see. Yet, in spite of this complexity,
such a study of spherical distributions through their harmonics lets one
probe deeper into their fundamental properties and characterize their
behaviour in terms of their symmetries etc. in a much more general
model-free setting rather than by dealing with specific parametric models.
We might mention here a paper \cite{healy1998spherical} which considers
nonparametric deconvolution of spherical densities via spherical harmonics.
It is our goal in this paper to discuss such harmonic analysis of spherical
distributions in a general setting, and to provide model-free large-sample
tests for testing various symmetries that these distributions might enjoy.

Harmonics play an important role in geosciences where they serve as a
smoothing and interpolation device for noisy data sets. A good example of
this is the map of global heat flow presented in \cite{CHAPMAN1975}; see
also \cite{GentonSTA477} where spatio-temporal data on global climate is
presented through visual animations. In the case of gravity and magnetic
field etc., spherical harmonics are the solutions to differential equations
that govern the potentials and provide a deeper physical relationship to the
fields. More applications for such analysis include computational geometry, 
\cite{kazhdan2004reflective}, image processing \cite{sun19973d}, approximate
symmetries in a large data set \cite{korman2015probably}, \textquotedblleft
big data\textquotedblright\ analysis for different meteorological data sets 
\cite{de2015tracing}, etc.

The paper is organized as follows. We start with a brief introduction to
spherical harmonics, and some essential properties that will be needed in
the rest of the paper. We then discuss some parametric models for the unit
sphere that have been commonly used in the literature to model 3-dimensional
directional data. We do this by introducing a very broad class of spherical
distributions called the \textquotedblleft Generalized-Fisher-Bingham"
family, that encompasses many currently used parametric models. In Section
3, we provide representation of the mean direction and dispersion for such
spherical models, in terms of their harmonics. In Section 4, we consider
various types of symmetries that a spherical distribution might enjoy, such
as isotropy, antipodal symmetry etc. Sections 5-7 deal with nonparametric
inference viz. estimation of the mean, as well as omnibus large-sample tests
for these various forms of symmetry. The final section deals with a real
data example of sunspot activity.

For a slightly more practical and computational side of these issues, the
reader is referred to a companion paper \cite{TGy_RSJ_WB2017} together with
a MATLAB package:``3D-Directional Statistics, Simulation and Visualization"
by the authors, which provides various simulation techniques and
visualization tools for these spherical models.

\subsection{Spherical Harmonics and the unit Sphere--some basics}

Continuous functions on a compact set can usually be approximated uniformly
by an orthogonal system of basis functions. In particular, we consider the
unit sphere in 3 dimensions, labelled the 2D-Sphere $\mathbb{S}_{2}$ in $%
\mathbb{R}^{3}$, with co-latitude $\vartheta\in\left[ 0,\pi\right] $ and
longitude $\varphi\in\left[ 0,2\pi\right] $. Continuous functions on such a
sphere can be approximated uniformly by composition of \textit{Orthonormal
Spherical Harmonics}, (see e.g. \cite{Mueller1966}, Theorem 9). Such a basis
set for $\mathbb{S}_{2}$ is given by the complex-valued functions $\{Y_{\ell
}^{m}\left( \vartheta,\varphi\right) $, $\ell=0,1,2,\ldots$, \ \ $m=-\ell
,-\ell+1,\ldots-1,0,1,\ldots,\ell-1,\ell\}$ of \textbf{degree} $\ell$ and 
\textbf{order} $m$, defined by 
\begin{equation}
Y_{\ell}^{m}\left( \vartheta,\varphi\right) =\left( -1\right) ^{m}\sqrt{%
\frac{2\ell+1}{4\pi}\frac{\left( \ell-m\right) !}{\left( \ell+m\right) !}}%
P_{\ell}^{m}\left( \cos\vartheta\right) e^{im\varphi },\;\varphi\in\left[
0,2\pi\right] ,\;\vartheta\in\left[ 0,\pi\right] ,
\label{SpheriHarmOrthoNorm}
\end{equation}
where $P_{\ell}^{m}$ denotes \textit{associated normalized }\textbf{Legendre
function} \textit{of the first kind, }(\cite{Erdelyi11a}, 3.2)---see
Appendix \ref{LegendrePoly}. These $\left\{ Y_{\ell}^{m}\right\} $ are fully%
\textit{\ }normalized in the sense, for each $\ell$ and $m$ 
\begin{equation*}
\int_{0}^{2\pi}\int_{0}^{\pi}\left\vert Y_{\ell}^{m}\left( \vartheta
,\varphi\right) \right\vert ^{2}\sin\vartheta d\vartheta d\varphi=1.
\end{equation*}
\smallskip If $\underline{\widetilde{x}}\left( \vartheta,\varphi\right) $
denotes a point on the unit sphere $\mathbb{S}_{2}$, we shall use the
alternate notations $Y_{\ell}^{m}\left( \vartheta,\varphi\right) $ and $%
Y_{\ell}^{m}\left( \underline{\widetilde{x}}\right) $ interchangeably. In
particular, the first few spherical harmonics are given by 
\begin{align}
Y_{0}^{0} & =\frac{1}{\sqrt{4\pi}},  \label{Y_eleje} \\
Y_{1}^{-1} & =\sqrt{\frac{3}{8\pi}}\sin\vartheta e^{-i\varphi},\quad
Y_{1}^{0}=\sqrt{\frac{3}{4\pi}}\cos\vartheta,\quad Y_{1}^{1}=-\sqrt{\frac {3%
}{8\pi}}\sin\vartheta e^{i\varphi},  \notag \\
Y_{2}^{-2} & =\sqrt{\frac{15}{32\pi}}\sin^{2}\vartheta e^{-i2\varphi},\quad
Y_{2}^{-1}=\sqrt{\frac{15}{8\pi}}\cos\vartheta\sin\vartheta e^{-i\varphi
},\quad Y_{2}^{0}=\sqrt{\frac{5}{16\pi}}\left( 3\cos^{2}\vartheta-1\right) 
\notag \\
Y_{2}^{1} & =-\sqrt{\frac{15}{8\pi}}\cos\vartheta\sin\vartheta e^{i\varphi
},\quad Y_{2}^{2}=\sqrt{\frac{15}{32\pi}}\sin^{2}\vartheta e^{i2\varphi }. 
\notag
\end{align}
The spherical harmonics of order $0$ have a special form so that for each $%
\ell$ 
\begin{equation}
Y_{\ell}^{0}\left( \vartheta,\varphi\right) =\sqrt{\frac{2\ell+1}{4\pi}}%
P_{\ell}\left( \cos\vartheta\right) .  \label{Y2P}
\end{equation}
where $P_{\ell}$ denotes Legendre polynomial, see Equation (\ref%
{Legendre_poly}).

\smallskip\noindent\textbf{Real-valued spherical harmonics:} Spherical
harmonics are in general complex valued, because they depend on $%
e^{im\varphi }$ where $\varphi$ is the longitude. Clearly $e^{im\varphi}$ is
a complete orthogonal system which is equivalent to the real valued
sine-cosine system. Similarly we can define real spherical harmonic
functions (note the double subscripts) 
\begin{equation*}
Y_{\ell,m}=\left\{ 
\begin{array}{cc}
\frac{1}{\sqrt{2}}\left( Y_{\ell}^{m}+\left( -1\right)
^{m}Y_{\ell}^{-m}\right) & m>0 \\ 
Y_{\ell}^{0} & m=0 \\ 
\frac{1}{i\sqrt{2}}\left( Y_{\ell}^{-m}-\left( -1\right)
^{m}Y_{\ell}^{m}\right) & m<0%
\end{array}
\right. .
\end{equation*}
The harmonics for $m>0$ are said to be of the cosine type, while those for $%
m<0$ are of the sine type.

As in Fourier analysis, complex-valued harmonics can be re-expressed in
terms of real-valued harmonics and sometimes they provide considerably
simpler formulae. However, in most cases, we will keep using the
complex-valued spherical harmonics and develop the basic theory using them.

\smallskip We refer the reader to \cite{Varshalovich1988} and \cite%
{SteinWeiss} for a more detailed account of spherical harmonics.

\section{Probability distributions on a unit sphere and their harmonic
representation}

Let $\underline{X}\in\mathbb{R}^{3}$ be a random variable with continuous
density function $f$, with characteristic function $\Phi\left( \underline{%
\omega}\right) =\mathsf{E}e^{i\underline{\omega}\cdot \underline{X}}$, and
its inverse%
\begin{equation*}
f\left( \underline{x}\right) =\frac{1}{\left( 2\pi\right) ^{3}}\int_{\mathbb{%
R}^{3}}e^{-i\underline{\omega}\cdot\underline{x}}\Phi\left( \underline{\omega%
}\right) d\underline{\omega}.
\end{equation*}
Recall the \textbf{\ Rayleigh plane wave expansion } of the exponential $%
e^{-i\underline{\omega}\cdot\underline{x}}$ (see \cite{Abramowit12} 10.1.47)
that for any $\underline{\omega}$, $\underline{x}\in\mathbb{R}^{3}$, $%
\rho=\left\vert \underline{\omega}\right\vert $, $r=\left\vert \underline{x}%
\right\vert $, 
\begin{align}
e^{i\underline{\omega}\cdot\underline{x}} &
=4\pi\sum_{\ell=0}^{\infty}\sum_{m=-\ell}^{\ell}i^{\ell}j_{\ell}\left( \rho
r\right) Y_{\ell}^{m}\left( \underline{\widetilde{\omega}}\right)
^{\ast}Y_{\ell}^{m}\left( \underline{\widetilde{x}}\right)
\label{Rayleigh_wave} \\
& =\sum_{\ell=0}^{\infty}\left( 2\ell+1\right) i^{\ell}j_{\ell}\left( \rho
r\right) P_{\ell}\left( \underline{\omega}\cdot\underline{x}\right) ,  \notag
\end{align}
where the notation $\ast$ is for the transpose and conjugate for a matrix
and just the conjugate for a scalar and where $j_{\ell}\left( z\right) $ is
the \textbf{Spherical Bessel function} of the first kind (see \cite%
{Abramowit12}, 10.1.1)), and is related to the Bessel function of the first
kind $J_{\ell+1/2}$ by the equation 
\begin{equation}
j_{\ell}\left( z\right) =\sqrt{\frac{\pi}{2z}}J_{\ell+1/2}\left( z\right) .
\label{SphericalBessel}
\end{equation}

Utilizing this, we may write the density $f$ in terms of the spherical
Bessel functions $j_{\ell}$ and spherical harmonic basis functions $%
Y_{\ell}^{m}$ (see (\ref{SphericalBessel}) and (\ref{SpheriHarmOrthoNorm})).
Putting $r=\left\vert \underline{x}\right\vert $, $\underline{\widetilde{x}}=%
\underline{x}/r$, $\rho=\left\vert \underline{\omega}\right\vert $, $%
\underline{\widetilde{\omega}}=\underline{\omega}/\rho$, for vectors $%
\underline{x},\underline{\omega}\in\mathbb{R}^{3}$, and denoting $%
\Omega\left( d\underline{\widetilde{\omega}}\right) $ as the Lebesgue
element of surface area on $\mathbb{S}_{2}$, we have 
\begin{align*}
f\left( \underline{x}\right) & =\frac{1}{\left( 2\pi\right) ^{3}}\int_{%
\mathbb{S}_{2}}\int_{0}^{\infty}4\pi\sum_{\ell=0}^{\infty}\sum_{m=-\ell
}^{\ell}\left( -i\right) ^{\ell}j_{\ell}\left( \rho r\right)
Y_{\ell}^{m}\left( \underline{\widetilde{\omega}}\right)
^{\ast}Y_{\ell}^{m}\left( \underline{\widetilde{x}}\right) \Phi\left( \rho%
\underline{\widetilde{\omega }}\right) \rho^{2}d\rho\Omega\left( d\underline{%
\widetilde{\omega}}\right) \\
& =\frac{1}{2\pi^{2}}\sum_{\ell=0}^{\infty}\sum_{m=-\ell}^{\ell}b_{\ell}^{m}%
\left( r\right) Y_{\ell}^{m}\left( \underline{\widetilde{x}}\right) ,
\end{align*}
where 
\begin{equation*}
b_{\ell}^{m}\left( r\right) =\int_{\mathbb{S}_{2}}\int_{0}^{\infty}\left(
-i\right) ^{\ell}j_{\ell}\left( r\rho\right) Y_{\ell}^{m}\left( \underline{%
\widetilde{\omega}}\right) ^{\ast}\Phi\left( \rho \underline{\widetilde{%
\omega}}\right) \rho^{2}d\rho\Omega\left( d\underline{\widetilde{\omega}}%
\right) .
\end{equation*}

In this paper, we are concerned with spherical distributions with density $%
f\left( \underline{\widetilde{x}}\right) $ on the unit sphere $\mathbb{S}%
_{2} $, which is a compact set in $\mathbb{R}^{3}$. When the density $f$ is
concentrated on unit sphere $\mathbb{S}_{2}$, we will write $a_{\ell
}^{m}=b_{\ell}^{m}\left( 1\right) /2\pi^{2}$.

\begin{theorem}[\protect\cite{Mueller1966}]
Let the density function $f\left( \underline{\widetilde{x}}\right) $ on the
unit sphere $\mathbb{S}_{2}$ be continuous. Then it has the series expansion
in terms of spherical harmonics $Y_{\ell}^{m}$ viz. 
\begin{equation}
f\left( \underline{\widetilde{x}}\right)
=\sum_{\ell=0}^{\infty}\sum_{m=-\ell}^{\ell}a_{\ell}^{m}Y_{\ell}^{m}\left( 
\underline{\widetilde{x}}\right) ,  \label{Series_f_L}
\end{equation}
where the complex valued coefficients $\left\{ a_{\ell}^{m}\right\} $ given
by 
\begin{equation}
a_{\ell}^{m}=\int_{\mathbb{S}_{2}}f\left( \underline{\widetilde{x}}\right)
Y_{\ell}^{m\ast}\left( \underline{\widetilde{x}}\right) \Omega\left( d%
\underline{\widetilde{x}}\right) ,  \label{a_l_m}
\end{equation}
and satisfy $a_{\ell}^{m\ast}=\left( -1\right) ^{m}a_{\ell}^{-m}$ and the
series (\ref{Series_f_L}) converges uniformly to $f$.
\end{theorem}

Notice that the spherical harmonic $Y_{0}^{0}=1/\sqrt{4\pi}$, hence $%
a_{0}^{0}=1/\sqrt{4\pi}$ is the normalizing constant for $f$. The series
expansion (\ref{Series_f_L}) for a spherical density is analogous to the
Fourier series expansion of a circular density (see e.g. Equation (2.1.5) in 
\cite{jammalamadaka2001topics}).

One can replace $Y_{\ell}^{m}$ by real spherical harmonics $Y_{\ell,m}$ in (%
\ref{Series_f_L}) then the coefficients $a_{\ell,m}$ of $Y_{\ell,m}$ will be
given in terms of $a_{\ell}^{m}$ as follows 
\begin{equation*}
a_{\ell,m}=\left\{ 
\begin{array}{cc}
\frac{1}{\sqrt{2}}\left( a_{\ell}^{m}+\left( -1\right)
^{m}a_{\ell}^{-m}\right) & m>0 \\ 
a_{\ell}^{0} & m=0 \\ 
\frac{1}{i\sqrt{2}}\left( a_{\ell}^{-m}-\left( -1\right)
^{m}a_{\ell}^{m}\right) & m<0.%
\end{array}
\right. .
\end{equation*}

\subsection{Harmonics for some specific spherical models}

In this section, we provide harmonic representations for many commonly used
spherical models and a few new ones. While this section is comprehensive in
covering most known spherical models, such a harmonic analysis is similar to
what one would find as Fourier representations of existing circular models
in the literature (See e.g. Section 2.2 of \cite{jammalamadaka2001topics} or
Section 3.5 of \cite{mardia2009directional}). The first 6 examples are
directly functions of the spherical harmonics, while Example \ref{Ex_GFB} is
a very broad class of parametric densities that is introduced here and
called, the \textquotedblleft Generalized Fisher-Bingham". Examples 8-11 are
various special cases of this family, and are widely used in the literature. 

\begin{example}
\label{Ex_Uniform}Uniform / isotropic distribution $f\left( \underline{%
\widetilde{x}}\right) =1/4\pi$, then $\left( \vartheta ,\varphi\right) $
have joint probability density 
\begin{equation*}
f\left( \vartheta,\varphi\right) =\frac{1}{4\pi}\sin\vartheta.
\end{equation*}
\end{example}

\begin{example}
\label{Ex_BM}Brownian Motion distribution, (\cite{mardia2009directional}
p.181)%
\begin{equation}
f\left( \underline{\widetilde{x}};\underline{\widetilde{x}}_{0},\zeta\right)
=\sum_{\ell=0}^{\infty}e^{-\ell\left( \ell+1\right) /4\zeta}\frac{2\ell +1}{%
4\pi}P_{\ell}\left( \underline{\widetilde{x}}_{0}\cdot \underline{\widetilde{%
x}}\right) ,  \label{dens_BM}
\end{equation}
where $\zeta>0$, notice (\ref{Y2P}).
\end{example}

Another class of densities is obtained by taking 
\begin{equation}
f\left( \underline{\widetilde{x}}\right) =\left\vert \sum_{\ell=0}^{\infty
}\sum_{m=-\ell}^{\ell}b_{\ell}^{m}Y_{\ell}^{m}\left( \underline{\widetilde{x}%
}\right) \right\vert ^{2}  \label{densityQM_General}
\end{equation}
for given values $\{b_{\ell}^{m}\}$, and this can then be re-expressed in
the general form (\ref{Series_f_L}), where $a_{\ell}^{n}$ can be explicitly
obtained from the Clebsch-Gordan series (see (\ref{Clebsch_Gordan_series})).

Two special cases described in Examples \ref{EX_Ylm_compl_suqare} and \ref%
{EX_Ylm_real_suqare}, are of particular interest. In quantum mechanics, the $%
Y_{\ell,m}\left( \underline{\widetilde{x}}\right) ^{2}$ itself and the
probability density function (\ref{densityQM_General}) plays an important
role in modelling the hydrogen atom, for instance when $\sum_{\ell=0}^{%
\infty}\sum_{m=-\ell}^{\ell}\left\vert b_{\ell}^{m}\right\vert ^{2}=1$. As
can be seen from equations (\ref{SpheriHarmOrthoNorm}) and (\ref%
{Density_SphH}), the squared modulus $\left\vert Y_{\ell}^{m}\left( 
\underline{\widetilde{x}}\right) \right\vert ^{2}$ is also a density
function and serves as a model for rotational symmetric density function on
sphere, since it depends on $\cos\vartheta$ only.

\begin{example}
\label{EX_Ylm_compl_suqare}The density $\left\vert Y_{\ell}^{m}\right\vert
^{2}$ has a representation of the form (\ref{Series_f_L}), as 
\begin{equation}
\left\vert Y_{\ell}^{m}\right\vert ^{2}=\left( -1\right) ^{m}\frac{2\ell +1}{%
\sqrt{4\pi}}\sum_{h=0}^{\ell}\sqrt{\frac{1}{4h+1}}C_{\ell,0;%
\ell,0}^{2h,0}C_{\ell,m;\ell,-m}^{2h,0}Y_{2h}^{0},  \label{Ylm_abs_suqare}
\end{equation}
where $C_{\ell_{1},m_{1};\ell_{2},m_{2}}^{k,m}\ $are Clebsch-Gordan
coefficients (see (\ref{Clebsch_Gordan_series} and \cite{TGy_RSJ_WB2017}).
\end{example}

\begin{example}
\label{EX_Ylm_real_suqare}For $m>0$, the density $Y_{\ell,m}^{2}$ has the
representation of the form (\ref{Series_f_L}) as 
\begin{align}
Y_{\ell,m}^{2} & =\frac{2\ell+1}{2\sqrt{2\pi}}\sum_{h=m}^{\ell}\sqrt {\frac{1%
}{4h+1}}C_{\ell,0;\ell,0}^{2h,0}C_{\ell,m;\ell,m}^{2h,2m}Y_{2h\boldsymbol{,}%
2m}  \notag \\
& +\left( -1\right) ^{m}\frac{2\ell+1}{\sqrt{4\pi}}\sum_{h=0}^{\ell}\sqrt{%
\frac{1}{4h+1}}C_{\ell,0;\ell,0}^{2h,0}C_{\ell,m;\ell,-m}^{2h,0}Y_{2h%
\boldsymbol{,0}}.  \label{Ylm_real_suqare}
\end{align}
\end{example}

See Appendix \ref{Appendix_Proofs1} for a proof.

\begin{example}
Exponential family (see \cite{beran1979exponential} and \cite%
{watson1983statistics} p.82)%
\begin{equation}
f_{e}\left( \underline{\widetilde{x}}\right)
=\exp\sum_{\ell=0}^{\infty}\sum_{m=-\ell}^{\ell}c_{\ell}^{m}Y_{\ell}^{m}%
\left( \underline{\widetilde{x}}\right) ,  \label{Exp_family}
\end{equation}
where $c_{\ell}^{m\ast}=\left( -1\right) ^{m}c_{\ell}^{-m}$. The normalizing
constant corresponds to $c_{0}^{0}$, it depends on the rest of the
parameters $c_{\ell}^{m}$, since the integral of $f_{e}$ must be $1$. Beran (%
\cite{beran1979exponential}) uses exponential of orthonormal spherical
harmonics, not separated by degree. The likelihood of observations also
takes this form. \newline
\end{example}

\begin{example}
Exponential densities with rotational symmetry around the axis $\underline{%
\widetilde{x}}_{0}$ have the form 
\begin{equation}
f_{e}\left( \underline{\widetilde{x}};\underline{\widetilde{x}}_{0}\right)
=\exp\sum_{\ell=0}^{\infty}c_{\ell}P_{\ell}\left( \cos\gamma\right) ,
\label{exp_Legrange}
\end{equation}
where $\gamma$ is the angle between $\underline{\widetilde{x}}$ and $%
\underline{\widetilde{x}}_{0}$, i.e. $\cos\gamma=\underline{\widetilde{x}}%
\cdot\underline{\widetilde{x}}_{0}$. The density $f_{e}\left( \underline{%
\widetilde{x}}\right) $ has the series expansion (\ref{Series_f_L}), 
\begin{equation*}
f_{e}\left( \underline{\widetilde{x}};\underline{\widetilde{x}}_{0}\right)
=\sum_{\ell=0}^{\infty}f_{\ell}\frac{2\ell+1}{4\pi}P_{\ell}\left( \cos
\gamma\right) ,
\end{equation*}
with%
\begin{equation*}
f_{\ell}=2\pi\int_{-1}^{1}\widehat{f}_{e}\left( x\right) P_{\ell}\left(
x\right) dx,
\end{equation*}
where we set $f_{e}\left( \underline{\widetilde{x}};\underline{\widetilde{x}}%
_{0}\right) =f_{e}\left( \underline{\widetilde{x}}\cdot \underline{%
\widetilde{x}}_{0}\right) =f_{e}\left( \cos\gamma\right) =\widehat{f}%
_{e}\left( x\right) $, see Appendix \ref{Appendix_Proofs_vMF} for a proof of
Example \ref{Ex_vMF}. \newline
For instance, consider the Generalized von Mises distribution discussed in 
\cite{gatto2007generalized} defined by 
\begin{equation*}
f_{e}\left( \underline{\widetilde{x}};\underline{\widetilde{x}}_{0}\right)
=\exp\left( \kappa_{0}+\kappa_{1}\cos\gamma+\kappa_{2}\cos2\gamma\right) ,
\end{equation*}
which can be rewritten as%
\begin{equation*}
f_{e}\left( \underline{\widetilde{x}};\underline{\widetilde{x}}_{0}\right)
=\exp\left( c_{0}+c_{1}P_{1}\left( \cos\gamma\right) +c_{2}P_{2}\left(
\cos\gamma\right) \right) ,
\end{equation*}
using $\cos2\gamma=2\cos^{2}\gamma-1$. A more general orthogonal form (\cite%
{maksimov1967necessary}) is to take 
\begin{equation*}
f_{e}\left( \underline{\widetilde{x}};\underline{\widetilde{x}}_{0}\right)
=\exp\left( \sum_{\ell=0}^{\infty}\kappa_{\ell}\cos\ell\gamma\right) .
\end{equation*}
This comes directly from (\ref{exp_Legrange}) using 
\begin{align*}
\frac{1}{4\pi}\int_{\mathbb{S}_{2}}P_{\ell}\left( \cos\vartheta\right) \cos
k\vartheta\Omega\left( d\underline{\widetilde{x}}\right) & =\frac{1}{4\pi }%
\int_{0}^{2\pi}\int_{0}^{\pi}P_{\ell}\left( \cos\vartheta\right) \cos
k\vartheta\sin\vartheta d\vartheta d\varphi \\
& =\frac{1}{2}\int_{0}^{\pi}P_{\ell}\left( \cos\vartheta\right) \cos
k\vartheta\sin\vartheta d\vartheta \\
& =\frac{1}{2}\int_{-1}^{1}P_{\ell}\left( x\right) T_{k}\left( x\right) dx,
\end{align*}
where $T_{k}$ denotes the Chebyshev polynomial of degree $k$.
\end{example}

\begin{example}[Generalized Fisher-Bingham family of spherical distributions]

\label{Ex_GFB}A very broad class of distributions, which we shall call
Generalized Fisher-Bingham family ($\mathbf{GFB}$), has the following
density 
\begin{equation*}
f\left( \underline{\widetilde{x}};\underline{\widetilde{\mu}},\kappa
,A\right) \cong\exp\left( \kappa\underline{\widetilde{\mu}}\cdot \underline{%
\widetilde{x}}+\underline{\widetilde{x}}^{\top}A\underline{\widetilde{x}}%
\right) ,
\end{equation*}
where $\cong$ denotes equality up to a normalizing constant and where $A$ is
a symmetric 3 by 3 matrix. Matrix $A$ has the form $A=MZM^{\top}$, where $M=%
\left[ \widetilde{\underline{\mu}}_{1},\widetilde{\underline{\mu}}_{2},%
\widetilde{\underline{\mu}}_{3}\right] $ is an orthogonal matrix and $%
Z=diag\left( \zeta_{1},\zeta_{2},\zeta_{3}\right) $ is a diagonal matrix. To
avoid identifiability problems, it is necessary to impose some constraint,
and the usual restriction imposed is that 
\begin{equation}
tr\left( A\right) =\zeta_{1}+\zeta_{2}+\zeta_{3}=0.  \label{Assumption_zeta}
\end{equation}
Further assume $\widetilde{\underline{\mu}}=\widetilde{\underline{\mu}}_{3}$%
, so that the number of parameters is $6$. The resulting $\mathbf{GFB}_{6}$
model has a density given by, 
\begin{equation}
f_{6}\left( \underline{\widetilde{x}};\kappa,\zeta_{1},\zeta_{2},M\right)
\cong\exp\left( \kappa\widetilde{\underline{\mu}}_{3}\cdot \underline{%
\widetilde{x}}+\sum_{k=1}^{3}\zeta_{k}\left( \widetilde{\underline{\mu}}%
_{k}\cdot\underline{\widetilde{x}}\right) ^{2}\right) .  \label{Density_FB6}
\end{equation}
An alternative parametrization of model $\mathbf{GFB}_{6}$, when the
constraint $\zeta_{2}=-\zeta_{1}$ is imposed, provides an alternate form of
the density 
\begin{equation}
f_{6}\left( \underline{\widetilde{x}};\kappa,\beta,\gamma,M\right)
\cong\exp\left( \kappa\underline{\widetilde{\mu}}_{3}\cdot\underline{%
\widetilde{x}}+\gamma\left( \underline{\widetilde{\mu}}_{3}\cdot\underline{%
\widetilde{x}}\right) ^{2}+\beta\left( \left( \underline{\widetilde{\mu}}%
_{1}\cdot\underline{\widetilde{x}}\right) ^{2}-\left( \underline{\widetilde{%
\mu }}_{2}\cdot\underline{\widetilde{x}}\right) ^{2}\right) \right) ,
\label{FB6_param_G}
\end{equation}
where again vectors $\widetilde{\underline{\mu}}_{1}$, $\widetilde{%
\underline{\mu}}_{2}$ and $\widetilde{\underline{\mu}}_{3}$ constitute an
orthogonal system.
\end{example}

\smallskip The next 4 examples provide some important special cases of this
Generalized Fisher-Bingham family, that have been considered in the
literature:

\begin{example}
\label{Ex_vMF} We first consider the \textit{von Mises-Fisher} (Fisher,
Langevin) distribution (see \cite{fisher1953dispersion}, \cite%
{jammalamadaka2001topics}) which is obtained from (\ref{FB6_param_G}) by
setting both $\gamma=0$ and $\beta=0$. The normalizing constant here
involves the \textbf{\ }Modified Bessel function of the first kind (see \cite%
{Abramowit12},9.6.18) given by 
\begin{equation*}
I_{\nu}\left( z\right) =\frac{\left( \frac{1}{2}z\right) ^{\nu}}{\sqrt {\pi}%
\Gamma\left( \nu+1/2\right) }\int_{0}^{\pi}e^{\pm
z\cos\vartheta}\sin^{2\nu}\vartheta d\vartheta,
\end{equation*}
and the density%
\begin{equation*}
f\left( \underline{\widetilde{x}};\underline{\widetilde{\mu}},\kappa\right) =%
\frac{\sqrt{\kappa}}{\left( 2\pi\right) ^{3/2}I_{1/2}\left( \kappa\right) }%
\exp\left( \kappa\underline{\widetilde{\mu}}\cdot\underline{\widetilde{x}}%
\right) ,
\end{equation*}
where $\kappa\geq0$. Notice that $f$ is a particular case of $\mathbf{GFB}%
_{6}$ setting $\gamma=0$, and $\beta=0$, in (\ref{FB6_param_G}). The density
function $f$ on $\mathbb{S}_{2}$ depends on cosine of the angle, say $\gamma$%
, between $\underline{\widetilde{x}}$ and $\underline{\widetilde{\mu}}$,
i.e. $\underline{\widetilde{x}}\cdot\underline{\widetilde{\mu}}=\cos\left(
\gamma\right) $. We find the series expansion (\ref{Series_f_L}) as 
\begin{equation}
f\left( \underline{\widetilde{x}};\underline{\widetilde{\mu}},\kappa\right)
=\sum_{\ell=0}^{\infty}\sqrt{\frac{2\ell+1}{4\pi}}\frac{I_{\ell+1/2}\left(
\kappa\right) }{I_{1/2}\left( \kappa\right) }Y_{\ell}^{0}\left( \underline{%
\widetilde{x}}\cdot\underline{\widetilde{\mu}}\right) .
\label{Density_Fisher}
\end{equation}
See Appendix \ref{Appendix_Proofs_vMF} for the proof.
\end{example}

%

\begin{example}
The model proposed by Bingham (\cite{bingham1974antipodally}) which we label 
$\mathbf{GFB}_{5,B}$, is obtained by setting $\kappa=0$ in (\ref{FB6_param_G}%
) and has the form%
\begin{equation*}
f_{B}\left( \underline{\widetilde{x}};\beta,\gamma,M\right) \cong\exp\left(
\gamma\left( \widetilde{\underline{\mu}}_{3}\cdot\underline{\widetilde{x}}%
\right) ^{2}+\beta\left( \left( \widetilde{\underline{\mu}}_{1}\cdot%
\underline{\widetilde{x}}\right) ^{2}-\left( \widetilde{\underline{\mu }}%
_{2}\cdot\underline{\widetilde{x}}\right) ^{2}\right) \right) ,
\end{equation*}
with five parameters.
\end{example}

For this special case with $\kappa=0$, see \cite{bingham1974antipodally} and
the companion paper \cite{TGy_RSJ_WB2017} for a more detailed account and
its visualization and simulation.

\begin{example}
\label{Example_DW} The Dimroth-Watson (Watson) Distribution, \cite%
{watson1983statistics}, is obtained by setting $\kappa=0$, and $\beta=0$, in
(\ref{FB6_param_G}), and we then have 
\begin{equation}
f\left( \underline{\widetilde{x}};\underline{\widetilde{\mu}},\gamma\right) =%
\frac{1}{M\left( 1/2,3/2,\gamma\right) }\exp\left( \gamma\left( \underline{%
\widetilde{\mu}}\cdot\underline{\widetilde{x}}\right) ^{2}\right) ,
\label{Distr_DW}
\end{equation}
where $M\left( 1/2,3/2,\gamma\right) $ is Kummer's function, \cite%
{mardia2009directional} p.181. The series expansion (\ref{Series_f_L}) of
this density is 
\begin{align}
f\left( \underline{\widetilde{x}};\underline{\widetilde{\mu}},\gamma\right)
& =\sum_{\ell=0}^{\infty}c_{\ell}\sqrt{\frac{2\ell+1}{4\pi}}%
Y_{\ell}^{0}\left( \vartheta,\varphi\right)  \label{Density_Watson} \\
& =\sum_{\ell=0}^{\infty}c_{\ell}\frac{2\ell+1}{4\pi}P_{\ell}\left(
\cos\vartheta\right) ,  \notag
\end{align}
where $\vartheta=\arccos\left( \underline{\widetilde{\mu}}\cdot \underline{%
\widetilde{x}}\right) $. For odd indices $c_{2\ell+1}=0$, and for the even
indices%
\begin{equation}
c_{2\ell}=\frac{2\pi}{M\left( 1/2,3/2,\gamma\right) }\int_{-1}^{1}\exp\left(
\gamma y^{2}\right) P_{2\ell}\left( y\right) dy.  \label{c_el_DW}
\end{equation}
\end{example}

%

\begin{example}
The model proposed by Kent (\cite{kent1982fisher}) which we label $\mathbf{%
GFB}_{5,K}$, arises by setting $\gamma=0$ in (\ref{FB6_param_G}), and has
the form 
\begin{equation}
f_{5}\left( \underline{\widetilde{x}};\kappa,\beta,M\right) \cong\exp\left(
\kappa\widetilde{\underline{\mu}}_{3}\cdot\underline{\widetilde{x}}%
+\beta\left( \left( \widetilde{\underline{\mu}}_{1}\cdot \underline{%
\widetilde{x}}\right) ^{2}-\left( \widetilde{\underline{\mu}}_{2}\cdot%
\underline{\widetilde{x}}\right) ^{2}\right) \right) ,  \label{Density_FB5}
\end{equation}
which defines the five parameter model $\mathbf{GFB}_{5,K}$.
\end{example}

\section{Mean direction, Moment of Inertia}

In this section, the Mean direction and dispersion of any spherical model is
expressed in terms of the harmonics.

\subsection{Mean direction in terms of spherical harmonics}

The mean of a random variable $\widetilde{\underline{X}}$ is calculated by 
\begin{align}
\underline{\mu} & =\mathsf{E}\widetilde{\underline{X}}=\int_{\mathbb{S}_{2}}%
\underline{\widetilde{x}}f\left( \underline{\widetilde{x}}\right)
\Omega\left( d\underline{\widetilde{x}}\right)  \label{Mean_Dir} \\
& =\int_{\mathbb{S}_{2}}\underline{\widetilde{x}}%
\sum_{m=-1}^{1}a_{1}^{m}Y_{1}^{m}\left( \underline{\widetilde{x}}\right)
\Omega\left( d\underline{\widetilde{x}}\right) .  \notag
\end{align}
We express the entries of $\underline{\widetilde{x}}=\left( \sin\vartheta
\cos\varphi,\sin\vartheta\sin\varphi,\cos\vartheta\right) ^{\top}$ in terms
of spherical harmonics 
\begin{equation}
\widetilde{x}_{1}=\sqrt{\frac{2\pi}{3}}\left( Y_{1}^{-1}-Y_{1}^{1}\right)
,\quad\widetilde{x}_{2}=i\sqrt{\frac{2\pi}{3}}\left(
Y_{1}^{-1}+Y_{1}^{1}\right) ,\quad\widetilde{x}_{3}=\sqrt{\frac{4\pi}{3}}%
Y_{1}^{0},  \label{coord2Sph}
\end{equation}
note $Y_{\ell}^{m\ast}\left( \vartheta,\varphi\right) =\left( -1\right)
^{m}Y_{\ell}^{-m}\left( \vartheta,\varphi\right) $. Or in terms of real
spherical harmonics 
\begin{equation*}
\widetilde{x}_{1}=\sqrt{\frac{4\pi}{3}}Y_{1,1},\quad\widetilde{x}_{2}=\sqrt{%
\frac{4\pi}{3}}Y_{1,-1},\quad\widetilde{x}_{3}=\sqrt{\frac{4\pi}{3}}Y_{1,0}.
\end{equation*}
Using the orthogonality of spherical harmonics, we have the components of $%
\underline{\mu}$ 
\begin{align}
\mu_{1} & =\int_{\mathbb{S}_{2}}\widetilde{x}_{1}%
\sum_{m=-1}^{1}a_{1}^{m}Y_{1}^{m}\left( \underline{\widetilde{x}}\right)
\Omega\left( d\underline{\widetilde{x}}\right)  \notag \\
& =\int_{\mathbb{S}_{2}}\sqrt{\frac{2\pi}{3}}\left(
Y_{1}^{-1}-Y_{1}^{1}\right) \sum_{m=-1}^{1}a_{1}^{m}Y_{1}^{m}\left( 
\underline{\widetilde{x}}\right) \Omega\left( d\underline{\widetilde{x}}%
\right)  \notag \\
& =\sqrt{\frac{2\pi}{3}}\left( a_{1}^{-1}-a_{1}^{1}\right) =\sqrt {\frac{4\pi%
}{3}}a_{1,1},  \label{Mu1}
\end{align}
similarly 
\begin{align}
\mu_{2} & =i\sqrt{\frac{2\pi}{3}}\left( a_{1}^{1}+a_{1}^{-1}\right) =\sqrt{%
\frac{4\pi}{3}}a_{1,-1},  \label{Mu2} \\
\mu_{3} & =\sqrt{\frac{4\pi}{3}}a_{1}^{0}=\sqrt{\frac{4\pi}{3}}a_{1,0}. 
\notag
\end{align}
Observe that although $a_{\ell}^{m}$ is complex valued, $\mu_{k}$ is real
and one can write\newline
$\underline{\mu}=R\left( \sin\vartheta_{\mu}\cos
\varphi_{\mu},\sin\vartheta_{\mu}\sin\varphi_{\mu},\cos\vartheta_{\mu}%
\right) ^{\top}=R\widetilde{\underline{\mu}}$, where the \emph{resultant} $R$
is defined by 
\begin{align*}
R^{2} & =\frac{4\pi}{3}\left( \left\vert a_{1}^{-1}\right\vert
^{2}+\left\vert a_{1}^{0}\right\vert ^{2}+\left\vert a_{1}^{1}\right\vert
^{2}\right) \\
& =\frac{4\pi}{3}\left( a_{1,1}^{2}+a_{1,-1}^{2}+a_{1,0}^{2}\right) ,
\end{align*}
and $\widetilde{\underline{\mu}}\in\mathbb{S}_{2}$, is the mean direction.
We see that the mean direction depends on the first degree coefficients 
\textit{only}. Also 
\begin{equation*}
R^{2}=\mu_{1}^{2}+\mu_{2}^{2}+\mu_{3}^{2}\leq\int_{\mathbb{S}_{2}}\left( 
\widetilde{x}_{1}^{2}+\widetilde{x}_{2}^{2}+\widetilde{x}_{3}^{2}\right)
f\left( \underline{\widetilde{x}}\right) \Omega\left( d\underline{\widetilde{%
x}}\right) =1
\end{equation*}
so that 
\begin{equation*}
R\leq1.
\end{equation*}

\begin{example}[\textbf{Ex. \protect\ref{Ex_Uniform} contd.}]
If $f\left( \underline{\widetilde{x}}\right) $ is uniform then $\underline{%
\mu}=0$ and $R=0$.
\end{example}

If $\underline{\mu}=0$, then the mean direction is undefined. We can have $%
R=0$ for non-uniform densities as well because of symmetries, as the
following examples illustrate.

\begin{example}[\textbf{Ex. \protect\ref{Example_DW} contd.}]
Dimroth-Watson Distribution (\ref{Density_Watson}) has $\underline{\mu}=0$,
as well as $R=0$.
\end{example}

\begin{example}[\textbf{Ex. \protect\ref{EX_Ylm_compl_suqare}, \protect\ref%
{EX_Ylm_real_suqare} contd.}]
For both the densities $\left\vert Y_{\ell}^{m}\left( \underline{\widetilde{x%
}}\right) \right\vert ^{2}$ and $Y_{\ell,m}^{2}$, we have $R=0$, since there
is no linear term in Clebsch-Gordan series 
(see (\ref{Ylm_abs_suqare}) and (\ref{Ylm_real_suqare}) respectively).
\end{example}

Let us consider the $\mathbf{GFB}_{6}$ model with parametrization (\ref%
{FB6_param_G}).

\begin{lemma}
The mean direction $\underline{\widetilde{\mu}}$ of the model $\mathbf{GFB}%
_{6}$ is characterized by a constant times $\kappa M^{\top}\underline{%
\widetilde{N}}$, in particular if $\kappa=0$, then $\underline{\widetilde{\mu%
}}$ is undefined ($\underline{\mu}=0$).
\end{lemma}

\begin{proof}
The mean direction 
\begin{align*}
\underline{\widetilde{\mu}} & =\int_{\mathbb{S}_{2}}\underline{\widetilde{x}}%
f_{6}\left( \underline{\widetilde{x}};\kappa,\beta,\gamma,M\right)
\Omega\left( d\underline{\widetilde{x}}\right) \\
& =M^{\top}\int_{\mathbb{S}_{2}}\underline{\widetilde{y}}f_{6}\left( 
\underline{\widetilde{y}};\kappa,\beta,\gamma\right) \Omega\left( d%
\underline{\widetilde{y}}\right) ,
\end{align*}
hence we shall consider the mean direction of the density in canonical form.
Rewriting the density $f_{6}$ in terms of real spherical harmonics 
\begin{align*}
f_{6}\left( \underline{\widetilde{y}};\kappa,\beta,\gamma\right) & \cong
e^{\kappa\widetilde{y}_{3}+\gamma\widetilde{y}_{3}^{2}+\beta\left( 
\widetilde{y}_{1}^{2}-\widetilde{y}_{2}^{2}\right) } \\
& =\exp\left( \kappa\sqrt{\frac{4\pi}{3}}Y_{1,0}+\gamma\sqrt{\frac{16\pi}{5}}%
Y_{2,0}+\frac{\sqrt{4\pi}}{3}Y_{0,0}+2\beta\sqrt{\frac{4\pi}{15}}%
Y_{2,2}\right) ,
\end{align*}
so that, 
\begin{equation}
f_{6}\left( \underline{\widetilde{y}};\kappa,\beta,\gamma\right)
\cong\exp\left( \sqrt{\frac{4\pi}{15}}\left( \kappa\sqrt{5}Y_{1,0}+\gamma 2%
\sqrt{3}Y_{2,0}+2\beta Y_{2,2}\right) \right) ,  \label{Density_FB6_Sph}
\end{equation}
where, $Y_{1,0}=Y_{1}^{0}$, $Y_{2,0}=Y_{2}^{0}$ and $Y_{2,2}=\left(
Y_{2}^{2}+Y_{2}^{-2}\right) /\sqrt{2}=\sqrt{15/16\pi}\sin^{2}\vartheta\cos%
\left( 2\varphi\right) $. \newline
The formula (\ref{Clebsch_Gordan_series}) for the product of two spherical
harmonics shows that the order in the right side is fixed at $m_{1}-m_{2}$,
while the degree changes from $\left\vert \ell _{1}-\ell_{2}\right\vert $ to 
$\ell_{1}+\ell_{2}$. The series expansion of (\ref{Density_FB6_Sph})
contains various products of spherical harmonics $Y_{1,0}$, $Y_{2,0}$ and $%
Y_{2,2}$, each of them with even order. Therefore, the series expansion (\ref%
{Series_f_L}) of the density (\ref{Density_FB6_Sph}) does not contain
spherical harmonics with odd order. The orthogonality of spherical harmonics
implies that 
\begin{equation*}
\int_{\mathbb{S}_{2}}\widetilde{y}_{k}f_{6}\left( \underline{\widetilde{y}}%
;\kappa,\beta,\gamma\right) \Omega\left( d\underline{\widetilde{y}}\right)
=0,
\end{equation*}
for $k=1$, $2$. Since the coordinates $\widetilde{y}_{1}$, $\widetilde{y}%
_{2} $, and $\widetilde{y}_{3}$ in terms of real spherical harmonics write
as 
\begin{equation*}
\widetilde{y}_{1}=\sqrt{\frac{2\pi}{3}}\left( Y_{1}^{-1}-Y_{1}^{1}\right)
,\quad\widetilde{y}_{2}=i\sqrt{\frac{2\pi}{3}}\left(
Y_{1}^{-1}+Y_{1}^{1}\right) ,\quad\widetilde{y}_{3}=\sqrt{\frac{4\pi}{3}}%
Y_{1}^{0},
\end{equation*}
and $\widetilde{y}_{1}$ and $\widetilde{y}_{2}$ are with odd order. The only
even order coordinate is $\widetilde{y}_{3}$. The result is that the first
two coordinates of mean direction of the density (\ref{Density_FB6_Sph}) is
zero, i.e. is a constant times North pole $\underline{\widetilde{N}}$. 
\newline
If $\kappa=0$, then consider the series expansion (\ref{Series_f_L}) of the
density (\ref{Density_FB6_Sph}) and conclude that not only the coefficients
with odd order are zero but the coefficients with odd degree as well, since $%
C_{\ell_{1},0;\ell_{2},0}^{\ell,0}=0$, if $\ell_{1}+\ell_{2}+\ell=2j+1$,
where $j$ is an integer.
\end{proof}

\subsection{Moment of Inertia and the Variance-Covariance matrix in terms of
spherical harmonics}

in order to consider the asymptotic distribution of the mean, or for
estimating the rotational axes etc., we need to discuss second order
moments, i.e. the moment of inertia, or the variance-covariance matrix,
which we do below.

If we define the product 
\begin{equation*}
\underline{\widetilde{x}}\underline{\widetilde{x}}^{\top}=%
\begin{bmatrix}
\sin^{2}\vartheta\cos^{2}\varphi & \sin^{2}\vartheta\cos\varphi\sin\varphi & 
\cos\vartheta\sin\vartheta\cos\varphi \\ 
\sin\vartheta\sin\varphi & \sin^{2}\vartheta\sin^{2}\varphi & \cos
\vartheta\sin\vartheta\sin\varphi \\ 
\cos\vartheta\sin\vartheta\sin\varphi & \cos\vartheta\sin\vartheta\sin\varphi
& \cos^{2}\vartheta%
\end{bmatrix}
,
\end{equation*}
this can be expressed in terms of real spherical harmonics 
\begin{equation}
\underline{\widetilde{x}}\underline{\widetilde{x}}^{\top}=%
\begin{bmatrix}
\sqrt{\frac{4\pi}{15}}Y_{2,2}-\frac{1}{3}\sqrt{\frac{4\pi}{5}}Y_{2,0}+\frac{%
\sqrt{4\pi}}{3}Y_{0,0} & \sqrt{\frac{4\pi}{15}}Y_{2,-2} & \sqrt {\frac{4\pi}{%
15}}Y_{2,1} \\ 
\sqrt{\frac{4\pi}{15}}Y_{2,-2} & -\sqrt{\frac{4\pi}{15}}Y_{2,2}-\frac{1}{3}%
\sqrt{\frac{4\pi}{5}}Y_{2,0}+\frac{\sqrt{4\pi}}{3}Y_{0,0} & \sqrt {\frac{4\pi%
}{15}}Y_{2,-1} \\ 
\sqrt{\frac{4\pi}{15}}Y_{2,1} & \sqrt{\frac{4\pi}{15}}Y_{2,-1} & \frac{2}{3}%
\sqrt{\frac{4\pi}{5}}Y_{2,0}+\frac{\sqrt{4\pi}}{3}Y_{0,0}%
\end{bmatrix}
,  \label{matrix_xx}
\end{equation}
so that%
\begin{equation*}
\mathsf{E}\widetilde{\underline{X}}\widetilde{\underline{X}}^{\top}=%
\begin{bmatrix}
\sqrt{\frac{4\pi}{15}}a_{2,2}-\frac{1}{3}\sqrt{\frac{4\pi}{5}}a_{2,0}+\frac {%
1}{3} & \sqrt{\frac{4\pi}{15}}a_{2,-2} & \sqrt{\frac{4\pi}{15}}a_{2,1} \\ 
\sqrt{\frac{4\pi}{15}}a_{2,-2} & -\sqrt{\frac{4\pi}{15}}a_{2,2}-\frac{1}{3}%
\sqrt{\frac{4\pi}{5}}a_{2,0}+\frac{1}{3} & \sqrt{\frac{4\pi}{15}}a_{2,-1} \\ 
\sqrt{\frac{4\pi}{15}}a_{2,1} & \sqrt{\frac{4\pi}{15}}a_{2,-1} & \frac{2}{3}%
\sqrt{\frac{4\pi}{5}}a_{2,0}+\frac{1}{3}%
\end{bmatrix}
.
\end{equation*}
Moment of Inertia here depends on second degree coefficients and constant 
\textit{only}. Notice that the trace of $\underline{\widetilde{x}}\underline{%
\widetilde{x}}^{\top}$ is $1$, as well as for $\mathsf{E}\widetilde{%
\underline{X}}\widetilde{\underline{X}}^{\top}$. Variance-Covariance matrix
is given by%
\begin{align}
\limfunc{Var}\left( \widetilde{\underline{X}}\right) & =\mathsf{E}\widetilde{%
\underline{X}}\widetilde{\underline{X}}^{\top }-\mathsf{E}\widetilde{%
\underline{X}}\mathsf{E}\widetilde{\underline{X}}^{\top}  \notag \\
& =\sqrt{\frac{4\pi}{15}}\left[ 
\begin{array}{ccc}
a_{2,2} & a_{2,-2} & a_{2,1} \\ 
a_{2,-2} & -a_{2,2} & a_{2,-1} \\ 
a_{2,1} & a_{2,-1} & 0%
\end{array}
\right] -\frac{4\pi}{3}\underline{\widetilde{\mu}}\underline{\widetilde{\mu}}%
^{\top}-\frac{1}{3}\sqrt{\frac{4\pi}{5}}a_{2,0}\limfunc{Diag}\left(
1,1,-2\right) +\frac{1}{3}\limfunc{Diag}\left( 1,1,1\right) .
\label{Var_Cov_matrix}
\end{align}

\begin{example}[\textbf{Ex. \protect\ref{EX_Ylm_real_suqare} contd.}]
For the density $Y_{\ell,m}^{2}$, $m>0$, the series expansion (\ref%
{Ylm_real_suqare}) provides the following coefficients: $a_{0,0}=\frac{1}{%
\sqrt{4\pi}}$, $a_{1,m}=0$, 
\begin{align*}
a_{2,0} & =\left( -1\right) ^{m}\frac{2\ell+1}{\sqrt{20\pi}}C_{\ell
,0;\ell,0}^{2,0}C_{\ell,m;\ell,-m}^{2,0}, \\
a_{2,2} & =\delta_{m=\boldsymbol{1}}\frac{2\ell+1}{\sqrt{40\pi}}%
C_{\ell,0;\ell,0}^{2,0}C_{\ell,1;\ell,1}^{2,2},
\end{align*}
$a_{2,1}=a_{2,-1}=a_{2,-2}=0$, see \cite{Varshalovich1988}, 8.5.1, p.248,
hence $\func{Var}\left( \widetilde{\underline{X}}\right) $ is diagonal. In
particular, if $\widetilde{\underline{X}}$ is with density $Y_{3,2}^{2}$,
then 
\begin{equation*}
\func{Var}\left( \widetilde{\underline{X}}\right) =%
\begin{bmatrix}
\frac{1}{3} & 0 & 0 \\ 
0 & \frac{1}{3} & 0 \\ 
0 & 0 & \frac{1}{3}%
\end{bmatrix}
, 
\end{equation*}
since $C_{3,2;3,-2}^{2,0}=0$, \cite{Varshalovich1988} 8.5,2 (45) p.252.
\end{example}

\smallskip

\section{Rotations and Symmetries}

Symmetries of physical systems, in particular the rules of atomic
spectroscopy, conservation of angular momentum etc. motivate the
consideration of the group of rotations of a 2D-sphere, which forms a
noncommutative group called $SO\left( 3\right) $ (see \cite{wigner2012group}%
). We follow the usual notation for a rotation $g\in SO\left( 3\right) $
acting on a function $f$ as $\Lambda\left( g\right) f\left( \underline{%
\widetilde{x}}\right) =f\left( g^{-1}\underline{\widetilde{x}}\right) $. In
particular $\Lambda\left( g\right) Y_{\ell}^{m}\left( \underline{\widetilde{x%
}}\right) =Y_{\ell}^{m}\left( g^{-1}\underline{\widetilde{x}}\right) $,
which is a rotated spherical harmonic, and is expressed in terms of
spherical harmonics in a natural manner in terms of Wigner D-matrices $%
D_{k,m}^{\left( \ell\right) }\left( g\right) $, i.e. 
\begin{equation}
\Lambda\left( g\right) Y_{\ell}^{m}\left( \underline{\widetilde{x}}\right)
=\sum_{k=-\ell}^{\ell}D_{k,m}^{\left( \ell\right) }\left( g\right)
Y_{\ell}^{k}\left( \underline{\widetilde{x}}\right) .  \label{MatrixWignerD}
\end{equation}
Applying such a rotation on a density function $f$ we have%
\begin{align*}
\Lambda\left( g\right) f\left( \underline{\widetilde{x}}\right) &
=\sum_{\ell=0}^{\infty}\sum_{m=-\ell}^{\ell}a_{\ell}^{m}\Lambda\left(
g\right) Y_{\ell}^{m}\left( \underline{\widetilde{x}}\right) \\
& =\sum_{\ell=0}^{\infty}\sum_{m=-\ell}^{\ell}a_{\ell}^{m}\sum_{k=-\ell
}^{\ell}D_{k,m}^{\left( \ell\right) }\left( g\right) Y_{\ell}^{k}\left( 
\underline{\widetilde{x}}\right) \\
& =\sum_{\ell=0}^{\infty}\sum_{k=-\ell}^{\ell}b_{\ell}^{k}\left( g\right)
Y_{\ell}^{k}\left( \underline{\widetilde{x}}\right) ,
\end{align*}
such that the new coefficients $b_{\ell}^{k}\left( g\right) $ are transforms
of $a_{\ell}^{m}$ as follows 
\begin{equation}
b_{\ell}^{k}\left( g\right)
=\sum_{m=-\ell}^{\ell}a_{\ell}^{m}D_{k,m}^{\left( \ell\right) }\left(
g\right) .  \label{Coeff_Rotation}
\end{equation}
We use this relationship for \textit{characterizing} symmetric densities in
particular cases.

\begin{remark}
In the case of an exponential family%
\begin{equation*}
\Lambda\left( g\right) f_{e}\left( \underline{\widetilde{x}}\right)
=\exp\sum_{\ell=0}^{\infty}\sum_{m=-\ell}^{\ell}c_{\ell}^{m}\Lambda\left(
g\right) Y_{\ell}^{m}\left( \underline{\widetilde{x}}\right) ,
\end{equation*}
so that the coefficients in the exponent after the rotation satisfy the same
equation (\ref{Coeff_Rotation}). The consequence is that all
characterizations of symmetry discussed below, follow directly from the
coefficients $c_{\ell }^{m}$.
\end{remark}

\section{Inference on isotropy on the sphere}

One of the central problems in dealing with directional data, prior to any
further inference, is to verify/test if the data is isotropic or uniformly
distributed over the space $\mathbb{S}_{2}$, in which case estimation of the
Mean or dispersion does not make sense.

\begin{definition}
$f$ is isotropic (globally symmetric), if for all $g\in SO\left( 3\right) $, 
$\Lambda\left( g\right) f\left( \underline{\widetilde{x}}\right) =f\left( 
\underline{\widetilde{x}}\right) $.
\end{definition}

This definition provides us a necessary and sufficient condition for the
uniformity which will be useful for inference.

\begin{lemma}
$f$ is globally symmetric if and only if $f$ is uniform distribution on the
sphere i.e. $f\left( \underline{\widetilde{x}}\right) =1/4\pi$.
\end{lemma}

\begin{proof}
We show that all coefficients in (\ref{Density_SphH}) are zero except $%
a_{0}^{0}$, i.e. $a_{\ell}^{m}=\delta_{\ell=0}\delta_{m=0}a_{0}^{0}$. This
follows from the equality of $b_{\ell}^{k}\left( g\right) $ and $a_{\ell
}^{k}$, by (\ref{Coeff_Rotation}), since integrating both sides of (\ref%
{Coeff_Rotation}) by the Haar measure $dg$ on $SO\left( 3\right) $ and using
(\ref{Int_D}), we get 
\begin{align*}
\int_{SO\left( 3\right) }a_{\ell}^{k}dg & =\sum_{m=-\ell}^{\ell}a_{\ell
}^{m}\int_{SO\left( 3\right) }D_{k,m}^{\left( \ell\right) }\left( g\right)
dg, \\
a_{\ell}^{k} & =\sum_{m=-\ell}^{\ell}a_{\ell}^{m}\delta_{\ell=0}\delta
_{m=0}\delta_{k=0}=a_{0}^{0}\delta_{\ell=0}\delta_{k=0}.
\end{align*}
\end{proof}

In order to develop an omnibus test of uniformity, we need to consider
estimation of the coefficients $a_{\ell}^{k}$.

\subsection{Estimation of $a_{\ell}^{m}$}

Consider now a random sample (i.i.d. observations) $\underline{\widetilde{x}}%
_{1}\left( \vartheta_{1},\varphi_{1}\right) $, $\underline{\widetilde{x}}%
_{2}\left( \vartheta_{2},\varphi_{2}\right) $,\ldots ,$\underline{\widetilde{%
x}}_{n}\left( \vartheta_{n},\varphi_{n}\right) $. Let the empirical density
function $\overline{f}_{n}\left( \underline{\widetilde{x}}\right) $ be
defined as usual by putting mass $1/n$ at each observation $\underline{%
\widetilde{x}}_{k}$. The estimator $\widehat{a_{\ell}^{m}}$, $\ell\neq0$, of 
$a_{\ell}^{m}$ given in (\ref{a_l_m}) be 
\begin{align}
\widehat{a_{\ell}^{m}} & =\int_{\mathbb{S}_{2}}\overline{f}_{n}\left( 
\underline{\widetilde{x}}\right) Y_{\ell}^{m\ast}\left( \underline{%
\widetilde{x}}\right) \Omega\left( d\underline{\widetilde{x}}\right)  \notag
\\
& =\frac{1}{n}\sum_{k=1}^{n}Y_{\ell}^{m\ast}\left( \underline{\widetilde{x}}%
_{k}\right) .  \label{Esti_a_l_m}
\end{align}
Clearly $\widehat{a_{\ell}^{m}}$ is unbiased since 
\begin{equation*}
\mathsf{E}\widehat{a_{\ell}^{m}}=\frac{1}{n}\sum_{k=1}^{n}\int_{\mathbb{S}%
_{2}}f\left( \underline{\widetilde{x}}\right) Y_{\ell}^{m\ast}\left( 
\underline{\widetilde{x}}\right) \Omega\left( d\underline{\widetilde{x}}%
\right) =a_{\ell}^{m},
\end{equation*}
and has variance%
\begin{equation*}
\limfunc{Var}\widehat{a_{\ell}^{m}}=\frac{1}{n}\left( \mathsf{E}\left\vert
Y_{\ell}^{m}\left( \underline{\widetilde{X}}\right) \right\vert
^{2}-\left\vert a_{\ell}^{m}\right\vert ^{2}\right)
\end{equation*}
where%
\begin{equation*}
\mathsf{E}\left\vert Y_{\ell}^{m}\left( \underline{\widetilde{X}}\right)
\right\vert ^{2}=\left( -1\right) ^{m}\frac{2\ell+1}{\sqrt{4\pi}}\sum_{0\leq
h\leq\ell}\sqrt{\frac{1}{4h+1}}C_{\ell,0;\ell,0}^{2h,0}C_{\ell,m;\ell
,-m}^{2h,0}a_{2h}^{0}.
\end{equation*}
It also follows that $\widehat{a_{\ell}^{m}}$ is consistent.

From now on we introduce the notation 
\begin{equation*}
C\left( \ell_{1},m_{1};\ell_{2},m_{2}\right) =\mathsf{E}Y_{\ell_{1}}^{m_{1}}%
\left( \underline{\widetilde{X}}\right) Y_{\ell_{2}}^{m_{2}}\left( 
\underline{\widetilde{X}}\right) ^{\ast},
\end{equation*}
then 
\begin{equation}
\limfunc{Cov}\left( Y_{\ell_{1}}^{m_{1}}\left( \underline{\widetilde{X}}%
\right) ,Y_{\ell_{2}}^{m_{2}}\left( \underline{\widetilde{X}}\right) \right)
=C\left( \ell_{1},m_{1};\ell _{2},m_{2}\right)
-a_{\ell_{1}}^{m_{1}}a_{\ell_{2}}^{m_{2}\ast}  \label{Cov_Ylm}
\end{equation}
where we can express $C$ in terms of Clebsch-Gordan coefficients 
\begin{align}
C\left( \ell_{1},m_{1};\ell_{2},m_{2}\right) & =\left( -1\right)
^{m_{2}}\sum_{\left\vert \ell_{1}-\ell_{2}\right\vert \leq
h\leq\ell_{1}+\ell_{2}}\sqrt{\frac{\left( 2\ell_{1}+1\right) \left(
2\ell_{2}+1\right) }{4\pi\left( 2h+1\right) }}C_{\ell_{1},0;%
\ell_{2},0}^{h,0}C_{\ell_{1},m_{1};\ell_{2},-m_{2}}^{h,m_{1}-m_{2}}\mathsf{E}%
Y_{h}^{m_{1}-m_{2}}\left( \underline{\widetilde{X}}\right)  \notag \\
& =\left( -1\right) ^{m_{2}}\sum_{\left\vert \ell_{1}-\ell_{2}\right\vert
\leq h\leq\ell_{1}+\ell_{2}}\sqrt{\frac{\left( 2\ell_{1}+1\right) \left(
2\ell_{2}+1\right) }{4\pi\left( 2h+1\right) }}C_{\ell_{1},0;%
\ell_{2},0}^{h,0}C_{\ell_{1},m_{1};%
\ell_{2},-m_{2}}^{h,m_{1}-m_{2}}a_{h}^{m_{1}-m_{2}}.  \label{Exp_Prod Y_gen}
\end{align}
In particular if we fix the degree $\ell$ then $C_{\ell,0;\ell,0}^{h,0}=0$
if $2\ell+h$ is odd, (\cite{Varshalovich1988}, 8.5, (h), p.250) 
\begin{align*}
C\left( \ell,m_{1};\ell,m_{2}\right) & =\left( -1\right) ^{m_{2}}\frac{%
2\ell+1}{\sqrt{4\pi}}\sum_{0\leq h\leq2\ell}\sqrt{\frac{1}{2h+1}}%
C_{\ell,0;\ell,0}^{h,0}C_{\ell,m_{1};%
\ell,-m_{2}}^{h,m_{1}-m_{2}}a_{h}^{m_{1}-m_{2}} \\
& =\left( -1\right) ^{m_{2}}\frac{2\ell+1}{\sqrt{4\pi}}\sum_{0\leq h\leq\ell}%
\sqrt{\frac{1}{4h+1}}C_{\ell,0;\ell,0}^{2h,0}C_{\ell,m_{1};%
\ell,-m_{2}}^{2h,m_{1}-m_{2}}a_{2h}^{m_{1}-m_{2}},
\end{align*}
hence $\limfunc{Cov}\left( Y_{\ell}^{m_{1}}\left( \underline{\widetilde{X}}%
\right) ,Y_{\ell}^{m_{2}}\left( \underline{\widetilde{X}}\right) \right) $
depends on coefficients $a_{2h}^{m_{1}-m_{2}}$, with even degrees $%
2h\leq2\ell$ and $a_{\ell}^{m_{1}}a_{\ell}^{m_{2}\ast}$.

Let $\underline{A}_{L}$ denote the theoretical vector of coefficients $%
a_{\ell}^{m}$, $\ell=1,2\ldots L$, $m=-\ell,\ldots,\ell$. Let us introduce
the corresponding vector of estimated coefficients $\widehat{\underline{A}%
_{L}}\left( n\right) =\left[ \widehat{\underline{a}_{1}}\left( n\right) ,%
\widehat{\underline{a}_{2}}\left( n\right) ,\ldots,\widehat{\underline{a}_{L}%
}\left( n\right) \right] ^{\top}$, where $\widehat{\underline{a}_{\ell }}%
\left( n\right) =\left[ \widehat{a_{\ell}^{-\ell}},\widehat{a_{\ell
}^{-\ell+1}},\ldots,\widehat{a_{\ell}^{\ell}}\right] $ are the estimated
coefficients of degree $\ell$. $\widehat{\underline{A}_{L}}\left( n\right) $
has dimension $L\left( L+2\right) $, and is in general complex valued except
for the entries $\widehat{a_{\ell}^{0}}$, $\ell=1,\ldots,L$.

\begin{theorem}
\label{Theor_AssNormality} $\widehat{\underline{A}_{L}}\left( n\right) $ is
asymptotically complex Gaussian ($\mathcal{CN}$) with expected value $%
\underline{A}_{L}$ and covariance matrix given by (\ref{Cov_Ylm}), i.e. 
\begin{equation*}
\sqrt{n}\left( \widehat{\underline{A}_{L}}\left( n\right) -\underline{A}%
_{L}\right) \overset{\mathcal{D}}{\rightarrow}\mathcal{CN}\left( 0,\mathbf{C}%
\right) ,
\end{equation*}
where $\mathbf{C}$ is as defined in (\ref{Cov_Ylm}).
\end{theorem}

\subsubsection{Estimation of the Mean direction}

Estimating the mean direction $\underline{\widetilde{\mu}}$ is very
important in order to construct several tests discussed below. The natural
estimator for $\underline{\widetilde{\mu}}$ based on a sample $\underline{%
\widetilde{x}}_{1}\left( \vartheta_{1},\varphi_{1}\right) $, $\underline{%
\widetilde{x}}_{2}\left( \vartheta_{2},\varphi_{2}\right) $, \ldots, $%
\underline{\widetilde{x}}_{n}\left( \vartheta_{n},\varphi_{n}\right) $ is 
\begin{equation*}
\widehat{\underline{\mu}}=\overline{\underline{\widetilde{x}}}=\frac{1}{n}%
\sum_{k=1}^{n}\underline{\widetilde{x}}_{k}\left( \vartheta_{k},\varphi
_{k}\right) .
\end{equation*}
It is easy to see that this estimator is equivalent to the one based on
estimating $a_{1}^{m}$ by $\widehat{a_{1}^{m}}$ first then using formula (%
\ref{Mean_Dir}) for estimator $\widehat{\underline{\mu}}$. The variance of
this estimator is given in terms of variance-covariance matrix (\ref%
{Var_Cov_matrix}). The asymptotic normality also follows as special cases of
Theorem \ref{Theor_AssNormality}. Estimating the variance-covariance matrix
is very similar and straightforward, see (\ref{Var_Cov_matrix}).

\subsection{Testing uniformity or global symmetry}

Testing uniformity is one of the well studied problem in circular
statistics, and there are a large number of tests (See e.g. Chapter 6 of 
\cite{jammalamadaka2001topics}). More recently Jammalamadaka et al. (\cite%
{RaoJammalamadaka2017}) provide a large sample test using the Fourier
coefficients for testing isotropy of circular data. In the case of the
sphere, as we have seen before, the density $f$ is globally symmetric or
isotropic if and only if it has the form $f\left( \underline{\widetilde{x}}%
\right) =1/4\pi$ --- in other words all coefficients $a_{\ell}^{m}$ are zero
in the series expansion (\ref{Series_f_L}) except $a_{0}^{0}=1/\sqrt{4\pi}$.
Now consider the null hypothesis

\begin{center}
$H_{0}$: $a_{\ell}^{m}=0$, if $\ell\neq0$, for all $m$.
\end{center}

Under this hypothesis, $\mathsf{E}\widehat{a_{\ell}^{m}}=0$, with variance 
\begin{equation*}
\limfunc{Var}\left( \widehat{a_{\ell}^{m}}\right) =\mathsf{E}\left\vert 
\widehat{a_{\ell}^{m}}\right\vert ^{2}-\left\vert \mathsf{E}\widehat{%
a_{\ell}^{m}}\right\vert ^{2}=\frac{1}{4\pi n}.
\end{equation*}
This follows from (\ref{Density_SphH}) and from the Clebsch-Gordan series (%
\ref{Ylm_abs_suqare}) for $\left\vert Y_{\ell}^{m}\left( \underline{%
\widetilde{x}}_{j}\right) \right\vert ^{2}$, so that 
\begin{equation*}
\mathsf{E}\left\vert Y_{\ell}^{m}\left( \underline{\widetilde{x}}_{j}\right)
\right\vert ^{2}=\frac{1}{\sqrt{4\pi}}Y_{0}^{0}=\frac{1}{4\pi},
\end{equation*}
further we use (\ref{Exp_Prod Y_gen}) and (\ref{Cl_Gord_00}) and have 
\begin{equation*}
\mathsf{E}Y_{\ell}^{m\ast}\left( \underline{\widetilde{x}}_{k}\right)
Y_{\ell}^{m}\left( \underline{\widetilde{x}}_{j}\right) =\delta_{j=k}\frac{1%
}{4\pi}+\delta_{j=k}\left\vert a_{\ell}^{m}\right\vert ^{2}=\delta_{j=k}%
\frac{1}{4\pi},
\end{equation*}
therefore%
\begin{equation*}
\mathsf{E}\left\vert \widehat{a_{\ell}^{m}}\right\vert ^{2}=\frac{1}{n^{2}}%
\sum_{k,j=1}^{n}\mathsf{E}Y_{\ell}^{m\ast}\left( \underline{\widetilde{x}}%
_{k}\right) Y_{\ell}^{m}\left( \underline{\widetilde{x}}_{j}\right) =\frac{1%
}{4\pi n}.
\end{equation*}
In general, under $H_{0}$, if $\ell_{1}\ell_{2}\neq0$, 
\begin{align*}
\limfunc{Cov}\left( Y_{\ell_{1}}^{m_{1}}\left( \underline{\widetilde{x}}%
_{k}\right) ,Y_{\ell_{2}}^{m_{2}}\left( \underline{\widetilde{x}}_{k}\right)
\right) & =\mathsf{E}Y_{\ell_{1}}^{m_{1}}\left( \underline{\widetilde{x}}%
_{k}\right) Y_{\ell_{2}}^{m_{2}\ast}\left( \underline{\widetilde{x}}%
_{k}\right) \\
& =\left( -1\right) ^{m_{2}}\sqrt{\frac{\left( 2\ell_{1}+1\right) \left(
2\ell_{2}+1\right) }{4\pi}}C_{\ell_{1},0;\ell_{2},0}^{\boldsymbol{0}%
,0}C_{\ell_{1},m_{1};\ell_{2},-m_{2}}^{\boldsymbol{0},0}Y_{0}^{0} \\
& =\left( -1\right) ^{m_{2}}\frac{2\ell_{1}+1}{\sqrt{4\pi}}\left( -1\right)
^{\ell_{1}}\delta_{\ell_{1}=\ell_{2}}\left( -1\right) ^{\ell
_{1}-m_{1}}\delta_{m_{1}=m_{2}}\frac{1}{2\ell_{1}+1}Y_{0}^{0} \\
& =\delta_{\ell_{1}=\ell_{2}}\delta_{m_{1}=m_{2}}\frac{1}{4\pi},
\end{align*}
while%
\begin{align*}
n^{2}Cov\left( \widehat{a_{\ell_{1}}^{m_{1}}},\widehat{a_{\ell_{2}}^{m_{2}}}%
\right) & =\sum_{k,\underline{\widetilde{x}}=1}^{n}\mathsf{E}%
Y_{\ell_{1}}^{m_{1}\ast}\left( \underline{\widetilde{x}}_{k}\right)
Y_{\ell_{2}}^{m_{2}}\left( \underline{\widetilde{x}}_{k}\right) \\
& =\sum_{k=1}^{n}\mathsf{E}Y_{\ell_{1}}^{m_{1}}\left( \underline{\widetilde{x%
}}_{k}\right) Y_{\ell_{2}}^{m_{2}\ast}\left( \underline{\widetilde{x}}%
_{k}\right) \\
& =n\delta_{\ell_{1}=\ell_{2}}\delta_{m_{1}=m_{2}}\frac{1}{4\pi}.
\end{align*}
Thus $\widehat{a_{\ell_{1}}^{m_{1}}}$ and $\widehat{a_{\ell_{2}}^{m_{2}}}$
are uncorrelated if $\ell_{1}\neq\ell_{2}$, $m_{1}\neq m_{2}$.

The real and imaginary parts of the complex Gaussian variate have the same
variance, which in our case is $1/8\pi n$ asymptotically. We have the same
value for $\left\vert \widehat{a_{\ell}^{m}}\right\vert =\left\vert \widehat{%
a_{\ell}^{-m}}\right\vert $, provided $m\neq0$, hence the asymptotic
distribution of $4\pi n\left( \left\vert \widehat{a_{\ell}^{m}}\right\vert
^{2}+\left\vert \widehat{a_{\ell}^{-m}}\right\vert ^{2}\right) $ is $\chi
^{2}$ with $2$ degrees of freedom, $4\pi n\left\vert \widehat{a_{\ell}^{0}}%
\right\vert ^{2}$ asymptotically $\chi^{2}$ with $1$ degree of freedom. It
follows that

\begin{theorem}
Under the above hypothesis $H_{0}$ we have 
\begin{equation*}
4\pi n\left\Vert \widehat{\underline{A}_{L}}\left( n\right) \right\Vert
^{2}=4\pi n\sum_{\ell=1}^{L}\sum_{m=-\ell}^{\ell}\left\vert \widehat{a_{\ell
}^{m}}\right\vert ^{2},
\end{equation*}
is asymptotically $\chi^{2}$ with $L\left( L+2\right) $ degrees of freedom.
\end{theorem}

This test contains the Rayleigh test for the special case of $L=1$, and
provides a general framework for testing uniformity.

\section{Inference on Rotational symmetry}

We consider rotational symmetry with respect to a given axis $\underline{%
\widetilde{x}}_{0}=\underline{\widetilde{x}}_{0}\left(
\vartheta_{0},\varphi_{0}\right) $.

\begin{definition}
$f$ is rotational symmetric about an axis $\underline{\widetilde{x}}_{0}$ if
for any rotation $g\in SO\left( 3\right) $, around axis $\underline{%
\widetilde{x}}_{0}$, $\Lambda\left( g\right) f\left( \underline{\widetilde{x}%
}\right) =f\left( \underline{\widetilde{x}}\right) $.
\end{definition}

We show a \textit{necessary and sufficient condition} for rotational
symmetry is that $f$ is a function of $\cos\gamma$ only, i.e. (\ref%
{Density_Rot_Sym}) be satisfied.

\begin{lemma}
$f$ is rotational symmetric about the axis $\underline{\widetilde{x}}_{0}$
if and only if \textbf{\ }$f\left( \underline{\widetilde{x}}\right) =f\left( 
\underline{\widetilde{x}}_{0}\cdot\widetilde{x}\right) =f\left( \cos
\gamma\right) $, where $\gamma$ is the angle between the rotation axis $%
\underline{\widetilde{x}}_{0}$ and the actual point $\underline{\widetilde{x}%
}$, moreover%
\begin{equation}
f\left( \cos\gamma\right) =\sum_{\ell=0}^{\infty}c_{\ell}\frac{2\ell+1}{4\pi}%
P_{\ell}\left( \cos\gamma\right) .  \label{Density_Rot_Sym}
\end{equation}
\end{lemma}

\begin{proof}
Assume rotational symmetry and rotate $\underline{\widetilde{x}}_{0}$ to the
coordinate $z$-axis by a rotation $g_{0}$. Now, let $f_{0}\left( \underline{%
\widetilde{x}}\right) =\Lambda\left( g_{0}\right) f\left( \underline{%
\widetilde{x}}\right) $. Consider a rotation about the coordinate $z$-axis
(axis of rotation is the North pole $\underline{\widetilde{N}}$) say by an
angle $\alpha$. This transforms the coefficients by (\ref{Coeff_Rotation}),
where $D_{k,m}^{\left( \ell\right) }\left( g\right)
=\delta_{m=k}e^{im\alpha} $ so that 
\begin{equation*}
a_{\ell}^{k}=\sum_{m=-\ell}^{\ell}a_{\ell}^{m}\delta_{m=k}e^{im\alpha}=a_{%
\ell}^{k}e^{ik\alpha},
\end{equation*}
(see \cite{Varshalovich1988}, 4.5.6 (30), p.84). Integrating over $%
\int_{0}^{2\pi}d\alpha$, we have $a_{\ell}^{k}=\delta_{k=0}a_{\ell}^{0}$.
Use $Y_{\ell}^{0}\left( \underline{\widetilde{x}}\right) =\sqrt{\left(
2\ell+1\right) /4\pi}P_{\ell}\left( \cos\vartheta\right) $, then we have 
\begin{align}
f_{0}\left( \underline{\widetilde{x}}\right) & =\sum_{\ell=0}^{\infty
}a_{\ell}^{0}Y_{\ell}^{0}\left( \underline{\widetilde{x}}\right)  \notag \\
& =\sum_{\ell=0}^{\infty}a_{\ell}^{0}\sqrt{\frac{2\ell+1}{4\pi}}P_{\ell
}\left( \cos\vartheta\right) .  \label{Rot_Sym}
\end{align}
Let us denote $g_{\underline{\widetilde{x}}_{0}}=g\left( \underline{\omega }%
;\vartheta_{0},\varphi_{0}\right) $ a rotation given in terms of Euler
angles $\underline{\omega}$ around $\underline{\widetilde{x}}_{0}$ by $%
\underline{\omega}$. Now we rotate back and forward $\underline{\widetilde{x}%
}_{0}$ to the North pole $\underline{\widetilde{N}}$ and get $%
a_{\ell}^{k}=\delta_{k=0}a_{\ell}^{0}$. To guarantee the series expansion by
an orthonormal system defined on $\left[ 0,\pi\right] $, instead of the
whole sphere (see remark below), we put the coefficient 
\begin{equation}
c_{\ell}=\sqrt{\frac{4\pi}{2\ell+1}}a_{\ell}^{0}  \label{c_el}
\end{equation}
in (\ref{Rot_Sym}) and obtain (\ref{Density_Rot_Sym}) where $\gamma$ is the
angle between $\underline{\widetilde{x}}$ and $\underline{\widetilde{x}}_{0}$%
. \newline
To show sufficiency, conversely assume \textbf{\ }$f\left( \underline{%
\widetilde{x}}\right) =f\left( \underline{\widetilde{x}}_{0}\cdot\widetilde{x%
}\right) =f\left( \cos\gamma\right) $. Observe $\cos \gamma=\underline{%
\widetilde{x}}_{0}\cdot\underline{\widetilde{x}}$, and in general if we have
the dependence of $f$ on $\underline{\widetilde{x}}_{0}\cdot\underline{%
\widetilde{x}}$, then the Funk-Hecke formula (see (\ref{Formula_Funk-Hecke}%
)) provides%
\begin{equation}
\int_{\mathbb{S}_{2}}f\left( \underline{\widetilde{x}}_{0}\cdot \underline{%
\widetilde{x}}\right) Y_{\ell}^{m\ast}\left( \underline{\widetilde{x}}%
\right) \Omega\left( d\underline{\widetilde{x}}\right)
=c_{\ell}Y_{\ell}^{m\ast}\left( \underline{\widetilde{x}}_{0}\right) ,
\label{Equation_rot_sym_Spheri}
\end{equation}
and 
\begin{equation*}
c_{\ell}=2\pi\int_{-1}^{1}f\left( y\right) P_{\ell}\left( y\right) dy,
\end{equation*}
i.e. $a_{\ell}^{m}=c_{\ell}Y_{\ell}^{m\ast}$, and 
\begin{align*}
f\left( \underline{\widetilde{x}}\right) &
=\sum_{\ell=0}^{\infty}\sum_{m=-\ell}^{\ell}a_{\ell}^{m}Y_{\ell}^{m}\left( 
\underline{\widetilde{x}}\right) \\
& =\sum_{\ell=0}^{\infty}c_{\ell}\sum_{m=-\ell}^{\ell}Y_{\ell}^{m}\left( 
\underline{\widetilde{x}}\right) Y_{\ell}^{m\ast}\left( \underline{%
\widetilde{x}}_{0}\right) \\
& =\sum_{\ell=0}^{\infty}c_{\ell}\frac{2\ell+1}{4\pi}P_{\ell}\left(
\cos\gamma\right) ,
\end{align*}
(see formula (\ref{FormulaAddition})). Now, if $f\left( \underline{%
\widetilde{x}}\cdot\underline{\widetilde{\mu}}\right) $ is a density of the
form (\ref{Density_Rot_Sym}) then 
\begin{equation*}
\int_{\mathbb{S}_{2}}f\left( \underline{\widetilde{x}}\cdot \underline{%
\widetilde{\mu}}\right) \Omega\left( d\underline{\widetilde{x}}\right) =%
\frac{c_{0}}{4\pi}4\pi=1,
\end{equation*}
hence $c_{0}=1$.
\end{proof}

\begin{remark}
Note here if $\underline{\widetilde{x}}_{0}$ \textbf{is not the North} pole,
then series expansion (\ref{Equation_rot_sym_Spheri}) is orthogonal in terms
of Legendre polynomials $P_{\ell}$, but \textbf{not } an orthogonal one in
the sense of (\ref{Series_f_L}) (see (\ref{Equation_rot_sym_Spheri})).
Notice if $\underline{\widetilde{x}}_{0}$ is not the North pole then 
\begin{equation*}
\sqrt{\frac{2\ell+1}{4\pi}}P_{\ell}\left( \cos\gamma\right) \neq Y_{\ell
}^{0}\left( \underline{\widetilde{x}}\right) =\sqrt{\frac{2\ell+1}{4\pi}}%
P_{\ell}\left( \cos\vartheta\right) .
\end{equation*}
For instance consider the density 
\begin{equation}
f\left( \underline{\widetilde{x}};\ \gamma\right) \cong e^{\gamma \widetilde{%
x}_{2}^{2}},  \label{Dens_Bingh_1}
\end{equation}
from the $\mathbf{GFB}$ family, where $\underline{\widetilde{x}}=\left( 
\widetilde{x}_{1},\widetilde{x}_{2},\widetilde{x}_{3}\right) ^{\top}$, hence 
\begin{equation*}
\widetilde{x}_{2}^{2}=-\sqrt{\frac{4\pi}{15}}Y_{2,2}-\frac{2}{3}\sqrt {\frac{%
\pi}{5}}Y_{2,0}+\frac{\sqrt{4\pi}}{3}Y_{0,0},
\end{equation*}
see (\ref{matrix_xx}), which depends not only on $\vartheta$ but on $\varphi$
as well. Now, let $\underline{\widetilde{y}}=\left( 0,1,0\right) ^{\top}$
the $y-$axis, therefore we have $\widetilde{x}_{2}^{2}=\left( \underline{%
\widetilde{y}}\cdot\underline{\widetilde{x}}\right) ^{2}=\cos ^{2}\gamma$,
and $f$ is rotationally symmetric with $\underline{\widetilde{x}}_{0}=%
\underline{\widetilde{y}}$.
\end{remark}

\subsubsection{Mean direction for a rotationally symmetric distribution.}

Writing the rotation matrix $M_{\mu}$ in terms of Euler angles such that $%
M_{\mu}\underline{\widetilde{\mu}}=\underline{\widetilde{N}}$, and putting $%
\underline{\widetilde{x}}=M_{\mu}^{-1}\underline{\widetilde{y}}$ 
\begin{align*}
\mathsf{E}\widetilde{\underline{X}} & =\int_{\mathbb{S}_{2}}\underline{%
\widetilde{x}}f\left( \underline{\widetilde{x}}\cdot \underline{\widetilde{%
\mu}}\right) \Omega\left( d\underline{\widetilde{x}}\right) \\
& =M_{\mu}^{-1}\int_{\mathbb{S}_{2}}\underline{\widetilde{y}}f\left( 
\underline{\widetilde{N}}\cdot\underline{\widetilde{y}}\right) \Omega\left( d%
\underline{\widetilde{y}}\right) \\
& =c_{1}M_{\mu}^{-1}\underline{\widetilde{N}}=c_{1}\underline{\widetilde{\mu 
}},
\end{align*}
since rotational axis of $f\left( \underline{\widetilde{N}}\cdot \underline{%
\widetilde{y}}\right) $ is $\underline{\widetilde{N}}$, $Y_{\ell
}^{0}=const.P_{\ell}$, and by (\ref{coord2Sph}) we have 
\begin{align*}
\int_{\mathbb{S}_{2}}\widetilde{y}_{k}P_{\ell}\left( \widetilde{y}%
_{3}\right) \Omega\left( d\underline{\widetilde{y}}\right) & =0,\quad k=1,2
\\
\int_{\mathbb{S}_{2}}\widetilde{y}_{3}P_{\ell}\left( \widetilde{y}%
_{3}\right) \Omega\left( d\underline{\widetilde{y}}\right) &
=\delta_{\ell=1}\int_{\mathbb{S}_{2}}P_{1}^{2}\left( \widetilde{y}\right)
\Omega\left( d\underline{\widetilde{y}}\right) \\
& =\delta_{\ell=1}\frac{4\pi}{3}.
\end{align*}
We see that the mean direction is $c_{1}\underline{\widetilde{\mu}}$, since
the rotational axis $\underline{\widetilde{\mu}}$ is a unit vector the
resultant is $R=\left\vert c_{1}\right\vert $.

\begin{example}[Ex. \protect\ref{Ex_vMF} contd.]
\textit{von Mises-Fisher} distribution (see (\ref{Density_Fisher})) has mean
direction%
\begin{equation*}
\mathsf{E}\widetilde{\underline{X}}=\frac{I_{3/2}\left( \kappa \right) }{%
I_{1/2}\left( \kappa \right) }\underline{\widetilde{\mu }}=\left( \coth
\left( \kappa \right) -\frac{1}{\kappa }\right) \underline{\widetilde{\mu }}.
\end{equation*}
\end{example}

\begin{example}[Ex. \protect\ref{Ex_BM} contd.]
Brownian Motion distribution (\ref{dens_BM}) has mean $\underline{\mu }=R%
\underline{\widetilde{\mu }}$, with resultant $R=e^{-1/2\zeta }$.
\end{example}

\begin{example}[Ex. \protect\ref{EX_Ylm_compl_suqare} contd.]
$\left\vert Y_{\ell }^{m}\left( \underline{\widetilde{x}}\right) \right\vert
^{2}$ is written in Clebsch-Gordan series, in terms of spherical harmonics
according to (\ref{Series_f_L}), (see (\ref{Ylm_abs_suqare})) 
\begin{equation}
\left\vert Y_{\ell }^{m}\left( \underline{\widetilde{x}}\right) \right\vert
^{2}=\left( -1\right) ^{m}\frac{2\ell +1}{4\pi }\sum_{h=0}^{\ell }C_{\ell
,0;\ell ,0}^{2h,0}C_{\ell ,m;\ell ,-m}^{2h,0}P_{2h}\left( \cos \vartheta
\right) ,  \label{Density_SphH}
\end{equation}%
therefore $\left\vert Y_{\ell }^{m}\left( \underline{\widetilde{x}}\right)
\right\vert ^{2}$ is rotational symmetric with axis $\underline{\widetilde{N}%
}$, and the resultant is $0$.
\end{example}

\subsection{Moment of Inertia for a rotationally symmetric distribution.}

As in the previous case \textbf{\ }%
\begin{align*}
\mathsf{E}\widetilde{\underline{X}}\widetilde{\underline{X}}^{\top} & =\int_{%
\mathbb{S}_{2}}\underline{\widetilde{x}}\underline{\widetilde{x}}%
^{\top}f\left( \underline{\widetilde{x}}\cdot\underline{\widetilde{\mu}}%
\right) \Omega\left( d\underline{\widetilde{x}}\right) \\
& =M_{\mu}^{-1}\int_{\mathbb{S}_{2}}^{\top}\underline{\widetilde{y}}%
\underline{\widetilde{y}}^{\top}f\left( \underline{\widetilde{N}}\cdot%
\underline{\widetilde{y}}\right) \Omega\left( d\underline{\widetilde{y}}%
\right) M_{\mu}
\end{align*}
Put $c_{\ell}\ $by (\ref{c_el}), in (\ref{Var_Cov_matrix}), then the inertia
matrix is 
\begin{equation*}
\mathsf{E}\widetilde{\underline{X}}\widetilde{\underline{X}}^{\top}=M_{\mu
}^{-1}%
\begin{bmatrix}
-\frac{1}{3}c_{2}+\frac{1}{3} & 0 & 0 \\ 
0 & -\frac{1}{3}c_{2}+\frac{1}{3} & 0 \\ 
0 & 0 & \frac{2}{3}c_{2}+\frac{1}{3}%
\end{bmatrix}
M_{\mu}.
\end{equation*}

\begin{example}[\textbf{Ex. \protect\ref{Example_DW} contd.}]
Dimroth-Watson distribution (see (\ref{Distr_DW})and (\ref{c_el_DW})) has
moment of inertia with 
\begin{equation*}
c_{2}=\frac{2\pi}{M\left( 1/2,3/2,\gamma\right) }\int_{-1}^{1}\exp\left(
\gamma y^{2}\right) P_{\ell}\left( y\right) dy. 
\end{equation*}
\end{example}

\begin{example}[\textbf{Ex. \protect\ref{Ex_BM} contd.}]
Brownian Motion distribution (\ref{dens_BM}) has moment of inertia with $%
c_{2}=e^{-3/2\zeta}$
\end{example}

\subsubsection{Estimation for rotationally symmetric distributions}

Let us consider the estimation of coefficients $c_{\ell}$ of rotational
symmetric density%
\begin{equation*}
f\left( \cos\vartheta\right) =\sum_{\ell=0}^{\infty}c_{\ell}\frac{2\ell +1}{%
4\pi}P_{\ell}\left( \cos\vartheta\right) .
\end{equation*}
It is straightforward considering 
\begin{equation*}
\widehat{c_{\ell}}=\frac{1}{n}\sum_{k=1}^{n}P_{\ell}\left( \widetilde{x}%
_{3,k}\right)
\end{equation*}
where $\widetilde{x}_{3,k}=\cos\vartheta_{k}$, is the third component of
sample $\underline{\widetilde{x}}_{k}$. Thus $\mathsf{E}\widehat{c_{\ell}}%
=c_{\ell}$, and for deriving the variance of $\widehat{c_{\ell}}$ we need an
expression for $P_{\ell}\left( \widetilde{x}_{3}\right) ^{2}$, we have%
\begin{align*}
P_{\ell}\left( \widetilde{x}_{3}\right) ^{2} & =\frac{4\pi}{2\ell +1}%
Y_{\ell}^{0}\left( \widetilde{x}_{3}\right) ^{2} \\
& =\sum_{0\leq h\leq\ell}\sqrt{\frac{4\pi}{4h+1}}\left( C_{\ell,0;\ell
,0}^{2h,0}\right) ^{2}Y_{2h}^{0}\left( \widetilde{x}_{3}\right) \\
& =\sum_{0\leq h\leq\ell}\left( C_{\ell,0;\ell,0}^{2h,0}\right)
^{2}P_{2h}\left( \widetilde{x}_{3}\right)
\end{align*}
compare to (\ref{E_Prod_Y_rot})%
\begin{align*}
n\limfunc{Var}\widehat{c_{\ell}} & =\limfunc{Var}P_{\ell}\left( \widetilde{X}%
_{3}\right) =\mathsf{E}P_{\ell}\left( \widetilde{X}_{3}\right)
^{2}-c_{\ell}^{2} \\
& =\sum_{0\leq h\leq\ell}\left( C_{\ell,0;\ell,0}^{2h,0}\right)
^{2}c_{2h}-c_{\ell}^{2}.
\end{align*}

\begin{example}
We consider the Dimroth-Watson Distribution ((\ref{Density_Watson})) with
parameter $\gamma=2$ and $\underline{\widetilde{\mu}}=\underline{\widetilde{N%
}}$. The coefficients $c_{\ell}$ are calculated according to Example \ref%
{Example_DW}. The random sample $\underline{\widetilde{x}}_{1}$, $\underline{%
\widetilde{x}}_{2}$, \ldots, $\underline{\widetilde{x}}_{n}$, $n=2^{12}$ was
simulated using MATLAB package 3D\_Directional\_SSV see \cite{TGy_RSJ_WB2017}
for details. The following table shows the calculated and estimated values
for he first four coefficients
\end{example}

\begin{center}
\begin{tabular}{|l|l|l|l|l|}
\hline
$\ell$ & $1$ & $2$ & $3$ & $4$ \\ \hline
$c_{\ell}$ & $0.0000$ & $0.2969$ & $0.0000$ & $0.0576$ \\ \hline
$\widehat{c_{\ell}}$ & $0.0052$ & $0.2982$ & $0.0042$ & $0.0596$ \\ \hline
$\limfunc{std}\left( P_{\ell}\left( \widetilde{X}_{3}\right) \right) $ & $%
0.7366$ & $0.4778$ & $0.4655$ & $0.4024$ \\ \hline
$\widehat{\limfunc{std}\left( P_{\ell}\left( \widetilde{x}_{3}\right)
\right) }$ & $0.7295$ & $0.4763$ & $0.4664$ & $0.4030$ \\ \hline
\end{tabular}
\end{center}

\subsection{Testing for rotational symmetry}

We now consider testing the null hypothesis that the data comes from a
distribution which is rotational symmetric around a \textit{given} axis,
which we can assume without loss of generality, is the North pole (by
rotating the specified axis to the North pole).

\begin{center}
$H_{0}$: $\underline{\widetilde{X}}$ is rotational symmetric around the
North pole $\underline{\widetilde{N}}$.
\end{center}

Under this hypothesis $f$ has the form (\ref{Density_Rot_Sym}). In other
words 
\begin{equation*}
\mathsf{E}Y_{\ell}^{m\ast}\left( \underline{\widetilde{X}}\right) =c_{\ell
}\delta_{m=0}\sqrt{\frac{2\ell+1}{4\pi}}.
\end{equation*}
Now, based on the observations $\underline{x}_{1}\left(
\vartheta_{1},\varphi_{1}\right) $, $\underline{x}_{2}\left(
\vartheta_{2},\varphi _{2}\right) $, \ldots, $\underline{x}_{n}\left(
\vartheta_{n},\varphi _{n}\right) $, we have that $\widehat{a_{\ell}^{m}}$
is unbiased $\mathsf{E}\widehat{a_{\ell}^{m}}=a_{\ell}^{m}$, and 
\begin{equation*}
\mathsf{E}\widehat{a_{\ell}^{m}}=\delta_{m=0}c_{\ell}\sqrt{\frac{2\ell+1}{%
4\pi}}.
\end{equation*}

For the variance, appealing again to the Clebsch-Gordan series (\ref%
{Exp_Prod Y_gen}), $m_{1}m_{2}\neq0$, we repeat $C\left(
\ell_{1},m_{1};\ell_{2},m_{2}\right) =\mathsf{E}Y_{\ell_{1}}^{m_{1}}\left( 
\underline{\widetilde{x}}_{j}\right) Y_{\ell_{2}}^{m_{2}\ast}\left( 
\underline{\widetilde{x}}_{j}\right) $, so 
\begin{align}
C\left( \ell_{1},m_{1};\ell_{2},m_{2}\right) & =\left( -1\right)
^{m_{2}}\sum_{\left\vert \ell_{1}-\ell_{2}\right\vert \leq k\boldsymbol{\leq 
}\ell_{1}+\ell_{2}}\sqrt{\frac{\left( 2\ell_{1}+1\right) \left( 2\ell
_{2}+1\right) }{4\pi\left( 2k+1\right) }}C_{\ell_{1},0;\ell_{2},0}^{k,0}C_{%
\ell_{1},m_{1};\ell_{2},-m_{2}}^{k,0}\sqrt{\frac{2k+1}{4\pi}}c_{k},  \notag
\\
& =\frac{\delta_{m_{1}=m_{2}}\left( -1\right) ^{m_{1}}}{4\pi}\sqrt{\left(
2\ell_{1}+1\right) \left( 2\ell_{2}+1\right) }\sum_{\left\vert \ell
_{1}-\ell_{2}\right\vert \leq
k\leq\ell_{1}+\ell_{2}}C_{\ell_{1},0;\ell_{2},0}^{k,0}C_{\ell_{1},m_{1};%
\ell_{2},-m_{1}}^{k,0}c_{k},  \label{E_Prod_Y_rot}
\end{align}
since $C_{\ell_{1},m_{1};\ell_{2},-m_{2}}^{k,0}=0$, unless $m_{1}-m_{2}=0$,
8,7,1 (2) p. 259. then for $m_{1}m_{2}\neq0$ 
\begin{equation*}
\limfunc{Cov}\left( Y_{\ell_{1}}^{m_{1}}\left( \underline{\widetilde{x}}%
_{j}\right) ,Y_{\ell_{2}}^{m_{2}\ast}\left( \underline{\widetilde{x}}%
_{j}\right) \right) =\delta_{m_{1}=m_{2}}C\left(
\ell_{1},m_{1};\ell_{2},m_{1}\right) .
\end{equation*}

Now we introduce the vector $\widehat{\underline{B}_{L}}\left( n\right) $ of
estimated coefficients, $\widehat{\underline{B}_{L}}\left( n\right) =\left[ 
\widehat{\underline{a}^{1}}\left( n\right) ,\widehat{\underline{a}^{2}}%
\left( n\right) ,\ldots,\widehat{\underline{a}^{L}}\left( n\right) \right]
^{\top}$, where $\widehat{\underline{a}^{m}}\left( n\right) =\left[ \widehat{%
a_{m}^{m}},\widehat{a_{m+1}^{m}},\ldots,\widehat{a_{L}^{m}}\right] ^{\top}$
with dimension $L-m+1$, are the coefficients with order $m$. $\widehat{%
\underline{B}_{L}}\left( n\right) $ is with dimension $L\left( L+1\right) /2$%
, it is complex valued. We have $\mathsf{E}\widehat{\underline{B}_{L}}\left(
n\right) =0$, and the covariance matrix of $\widehat{\underline{B}_{L}}%
\left( n\right) $, denote it $\mathcal{C}_{B,L}=\limfunc{Cov}\left( \widehat{%
\underline{B}_{L}}\left( n\right) ,\widehat{\underline{B}_{L}}\left(
n\right) \right) $ is block diagonal with blocks $\limfunc{Cov}\left( 
\widehat{\underline{a}^{m}}\left( n\right) ,\widehat{\underline{a}^{m}}%
\left( n\right) \right) $. The entries of $\limfunc{Cov}\left( \widehat{%
\underline{a}^{m}}\left( n\right) ,\widehat{\underline{a}^{m}}\left(
n\right) \right) $ are given in terms of $C\left(
\ell_{1},m;\ell_{2},m\right) $.

\begin{lemma}
If $\underline{\widetilde{X}}$ is rotational symmetric then $2n\widehat{%
\underline{B}_{L}}\left( n\right) ^{\ast}\mathcal{C}_{B,L}^{-1}\widehat{%
\underline{B}_{L}}\left( n\right) $ is $\chi_{L\left( L+1\right) }^{2}$
distributed with $L\left( L+1\right) $ degrees of freedom.
\end{lemma}

In particular, we have for $m\neq0$, 
\begin{align*}
\mathsf{E}\left\vert \widehat{a_{\ell}^{m}}\right\vert ^{2} & =\frac {1}{%
n^{2}}\sum_{k,j=1}^{n}\mathsf{E}Y_{\ell}^{m\ast}\left( \underline{\widetilde{%
x}}_{k}\right) Y_{\ell}^{m}\left( \underline{\widetilde{x}}_{j}\right) \\
& =\frac{1}{n}C\left( \ell,m;\ell,m\right) ,
\end{align*}
and 
\begin{align*}
C\left( \ell,\ell;\ell,\ell\right) & =n\mathsf{E}\left\vert \widehat{%
a_{\ell}^{\ell}}\right\vert ^{2}=n\mathsf{E}Y_{\ell}^{\ell\ast }\left( 
\underline{\widetilde{x}}_{j}\right) Y_{\ell}^{\ell}\left( \underline{%
\widetilde{x}}_{j}\right) \\
& =\frac{\left( -1\right) ^{\ell}}{4\pi}\left( 2\ell+1\right) \sum_{0\leq k%
\boldsymbol{\leq}\ell}C_{\ell,0;\ell,0}^{2k,0}C_{\ell,\ell;\ell,-\ell}^{2k,0}
\end{align*}
see 8.5.2 a (32), b. (36) p. 251. Hence $\widehat{a_{\ell}^{\ell}}$ are
uncorrelated $\ $and small, we reject the null hypothesis if 
\begin{equation*}
2n\sum_{\ell=1}^{L}\frac{\left\vert \widehat{a_{\ell}^{\ell}}\right\vert ^{2}%
}{C\left( \ell,\ell;\ell\right) }\sim\chi_{2L}^{2},
\end{equation*}
is sufficiently large. In this way we simplify the statistic at the cost of
smaller degrees of freedom.

\section{Other forms of symmetry}

\subsection{Axial or antipodal symmetry}

Antipodal or axial symmetry refers to the density being the same at opposite
ends of the diameter i.e. the density at the points $\underline{\widetilde{x}%
}\left( \vartheta,\varphi\right) $ and $\underline{\widetilde{x}}\left(
\pi-\vartheta,\varphi+\pi\right) $ match (see \cite{bingham1974antipodally}%
). Without loss of generality, we may assume that the normal is the North
pole, because otherwise one can rotate the normal to the North pole.

\begin{definition}
$f$ is axially / antipodally symmetric, if $f\left( \underline{\widetilde{x}}%
\right) =f\left( -\underline{\widetilde{x}}\right) $ for all $\underline{%
\widetilde{x}}\in\mathbb{S}_{2}$.
\end{definition}

For the \textquotedblleft inversion" $\underline{\widetilde{x}}\left(
\vartheta,\varphi\right) \rightarrow-\underline{\widetilde{x}}\left(
\pi-\vartheta,\pi+\varphi\right) $, we have 
\begin{equation}
Y_{\ell}^{m}\left( -\underline{\widetilde{x}}\right) =\left( -1\right)
^{\ell}Y_{\ell}^{m}\left( \underline{\widetilde{x}}\right) .
\label{FormulaInversion}
\end{equation}
If $f\left( \underline{\widetilde{x}}\right) =f\left( -\underline{\widetilde{%
x}}\right) $, then 
\begin{align*}
f\left( \underline{\widetilde{x}}\right) &
=\sum_{\ell=0}^{\infty}\sum_{m=-\ell}^{\ell}a_{\ell}^{m}Y_{\ell}^{m}\left( 
\underline{\widetilde{x}}\right) \\
& =f\left( -\underline{\widetilde{x}}\right) \\
& =\sum_{\ell=0}^{\infty}\sum_{m=-\ell}^{\ell}\left( -1\right) ^{\ell
}a_{\ell}^{m}Y_{\ell}^{m}\left( \underline{\widetilde{x}}\right) ,
\end{align*}
see (\ref{FormulaInversion}), therefore $a_{2\ell+1}^{m}=-a_{2\ell+1}^{m}$,
hence $a_{2\ell+1}^{m}=0$, for all $m$ and $\ell$. We have the following

\begin{lemma}
$f$ is axially symmetric if and only if for all $\ell$ and $m$, $a_{2\ell
+1}^{m}=0$, in this case 
\begin{equation*}
f\left( \underline{\widetilde{x}}\right)
=\sum_{\ell=0}^{\infty}\sum_{m=-2\ell}^{2\ell}a_{2\ell}^{m}Y_{2\ell}^{m}%
\left( \underline{\widetilde{x}}\right) .
\end{equation*}
\end{lemma}

In this case, the mean direction is undefined $\mathsf{E}\underline{%
\widetilde{X}}=0$.

\begin{example}[\textbf{Ex. \protect\ref{EX_Ylm_compl_suqare}, \protect\ref%
{EX_Ylm_real_suqare} contd.}]
Both $\left\vert Y_{\ell}^{m}\right\vert ^{2}$ (see (\ref{Ylm_abs_suqare}))
and $Y_{\ell,m}^{2}$ (see (\ref{Ylm_real_suqare}) are axially / antipodally
symmetric.
\end{example}

\subsubsection{Testing for Axial symmetry}

In this case as we have seen before, the mean direction $\mathsf{E}%
\underline{\widetilde{X}}=0$. For an observation vector $\underline{%
\widetilde{x}}_{1}\left( \vartheta_{1},\varphi_{1}\right) $, $\underline{%
\widetilde{x}}_{2}\left( \vartheta_{2},\varphi_{2}\right) $, \ldots, $%
\underline{\widetilde{x}}_{n}\left( \vartheta_{n},\varphi _{n}\right) $, we
have that $\widehat{a_{\ell}^{m}}$ is unbiased $\mathsf{E}\widehat{%
a_{\ell}^{m}}=a_{\ell}^{m}$, therefore 
\begin{equation*}
\mathsf{E}\widehat{a_{2\ell}^{m}}=a_{2\ell}^{m},\;\mathsf{E}\widehat{%
a_{2\ell +1}^{m}}=0.
\end{equation*}
Hence all $\widehat{a_{2\ell+1}^{m}}$ are small. We can calculate the
covariance matrix for estimators $\widehat{a_{2\ell+1}^{m}}$, consider first 
\begin{align*}
& \mathsf{E}Y_{2\ell_{1}+1}^{m_{1}}\left( \underline{\widetilde{X}}\right)
Y_{2\ell_{2}+1}^{m_{2}}\left( \underline{\widetilde{X}}\right) ^{\ast} \\
& =\left( -1\right) ^{m_{2}}\sum_{2\left\vert \ell_{1}-\ell_{2}\right\vert
\leq h\leq2\ell_{1}+2\ell_{2}+2}\sqrt{\frac{\left( 4\ell_{1}+3\right) \left(
4\ell_{2}+3\right) }{4\pi\left( 2h+1\right) }}C_{2\ell
_{1}+1,0;2\ell_{2}+1,0}^{h,0}C_{2\ell_{1}+1,m_{1};2%
\ell_{2}+1,-m_{2}}^{h,m_{1}-m_{2}}a_{h}^{m_{1}-m_{2}},
\end{align*}
by (\ref{Exp_Prod Y_gen}). Since $2\ell_{1}+2\ell_{2}+2$ is even, so $h$
must be even as well%
\begin{align*}
& \mathsf{E}Y_{2\ell_{1}+1}^{m_{1}}\left( \underline{\widetilde{X}}\right)
Y_{2\ell_{2}+1}^{m_{2}}\left( \underline{\widetilde{X}}\right) ^{\ast} \\
& =\left( -1\right) ^{m_{2}}\sum_{\left\vert \ell_{1}-\ell_{2}\right\vert
\leq h\leq\ell_{1}+\ell_{2}+1}\sqrt{\frac{\left( 4\ell_{1}+3\right) \left(
4\ell_{2}+3\right) }{4\pi\left( 4h+3\right) }}C_{2\ell_{1}+1,0;2\ell
_{2}+1,0}^{2h,0}C_{2\ell_{1}+1,m_{1};2%
\ell_{2}+1,-m_{2}}^{2h,m_{1}-m_{2}}a_{2h}^{m_{1}-m_{2}}.
\end{align*}

Let $\widehat{\underline{A}_{L}}\left( n\right) =\left[ \widehat{\underline{a%
}_{1}}\left( n\right) ,\widehat{\underline{a}_{3}}\left( n\right) ,\ldots,%
\widehat{\underline{a}_{2L+1}}\left( n\right) \right] ^{\top}$, where $%
\widehat{\underline{a}_{2\ell+1}}\left( n\right) =\left[ \widehat{%
a_{2\ell+1}^{-2\ell-1}},\widehat{a_{2\ell+1}^{-2\ell}},\ldots,\widehat{%
a_{2\ell+1}^{2\ell+1}}\right] $, and the covariance matrix $\mathcal{C}$ of $%
\widehat{\underline{A}_{L}}\left( n\right) $ is given in terms $\mathsf{E}%
Y_{2\ell_{1}+1}^{m_{1}}\left( \underline{\widetilde{X}}\right)
Y_{2\ell_{2}+1}^{m_{2}}\left( \underline{\widetilde{X}}\right) ^{\ast}$.

\begin{lemma}
If $\underline{\widetilde{X}}$ is axially symmetric then $2n\widehat{%
\underline{A}_{L}}\left( n\right) ^{\ast}\mathcal{C}^{-1}\widehat{\underline{%
A}_{L}}\left( n\right) $ has a $\chi_{\left( L+2\right) L}^{2}$ distribution
for large $n$.
\end{lemma}

We reject this null hypothesis of axial symmetry if $2n\widehat{\underline{A}%
_{L}}\left( n\right) ^{\ast}\mathcal{C}^{-1}\widehat{\underline{A}_{L}}%
\left( n\right) $ is large.

\subsection{Reflection with respect to equatorial plane}

This kind of symmetry may be observed e.g. in crystallography and
astrophysics.

\begin{definition}
$f$ is symmetric with respect to equatorial plane, if for any $\underline{%
\widetilde{x}},\underline{\widetilde{x}}^{\prime}\in\mathbb{S}_{2}$, such
that $\underline{\widetilde{x}}\left( \vartheta,\varphi\right) \rightarrow%
\underline{\widetilde{x}}^{\prime}\left( \pi-\vartheta ,\varphi\right) $,
then $f\left( \underline{\widetilde{x}}\right) =f\left( \underline{%
\widetilde{x}}^{\prime}\right) $.
\end{definition}

\begin{lemma}
\label{Lemma_Equa_plane}$f$ is symmetric with respect to equatorial plane if
and only if $a_{2\ell+1}^{2m}=0$, and $a_{2\ell}^{2m+1}=0$, for all $\ell$,
and $m$, in this case%
\begin{equation*}
f\left( \underline{\widetilde{x}}\right) =\sum_{\ell=0}^{\infty}\left(
\sum_{m=-\ell}^{\ell}a_{2\ell}^{2m}Y_{2\ell}^{2m}\left( \underline{%
\widetilde{x}}\right)
+\sum_{m=-\ell-1}^{\ell}a_{2\ell+1}^{2m+1}Y_{2\ell+1}^{2m+1}\left( 
\underline{\widetilde{x}}\right) \right) ,
\end{equation*}
\end{lemma}

See Appendix \ref{Appendix_Proofs2} for a proof.

\begin{remark}
Mean direction $\underline{\widetilde{\mu}}$ for a density which is
symmetric with respect to equatorial plane is $\mathsf{E}\underline{%
\widetilde{X}}=\sqrt{4\pi/3}\left( -a_{1,1},a_{1,-1},0\right) ^{\top}$, and
belongs to the $x,y$ plane.
\end{remark}

\subsubsection{Testing for symmetry with respect to the equatorial plane}

In this case, since the mean direction is given by $\mathsf{E}\underline{%
\widetilde{X}}=\sqrt{4\pi/3}\left( -a_{1,1},a_{1,-1},0\right) ^{\top}$, we
can clearly reject the hypothesis by testing if $\widehat{a_{1}^{0}}=0$.
Since this is not a sufficient condition to claim this type of symmetry, we
now consider the following more comprehensive procedure.

For an observation vector $\underline{\widetilde{x}}_{1}\left( \vartheta
_{1},\varphi_{1}\right) $, $\underline{\widetilde{x}}_{2}\left(
\vartheta_{2},\varphi_{2}\right) $, \ldots, $\underline{\widetilde{x}}%
_{n}\left( \vartheta_{n},\varphi_{n}\right) $, we have $\widehat{a_{\ell
}^{m}}$ is unbiased, therefore 
\begin{equation}
\mathsf{E}\widehat{a_{2\ell+1}^{2m}}=0,\;\mathsf{E}\widehat{a_{2\ell}^{2m+1}}%
=0,  \label{Equa_al}
\end{equation}
so that we expect all $\widehat{a_{2\ell+1}^{2m}}$ and $\widehat{a_{2\ell
}^{2m+1}}$ to be small. Again we can calculate the covariance matrix. By (%
\ref{Exp_Prod Y_gen}) we have 
\begin{align*}
& \mathsf{E}Y_{\ell_{1}}^{m_{1}}\left( \underline{\widetilde{X}}\right)
Y_{\ell_{2}}^{m_{2}}\left( \underline{\widetilde{X}}\right) ^{\ast} \\
& =\left( -1\right) ^{m_{2}}\sum_{\left\vert \ell_{1}-\ell_{2}\right\vert
\leq h\leq\ell_{1}+\ell_{2}}\sqrt{\frac{\left( 2\ell_{1}+1\right) \left(
2\ell_{2}+1\right) }{4\pi\left( 2h+1\right) }}C_{\ell_{1},0;%
\ell_{2},0}^{h,0}C_{\ell_{1},m_{1};%
\ell_{2},-m_{2}}^{h,m_{1}-m_{2}}a_{h}^{m_{1}-m_{2}}.
\end{align*}
Now if the parity of $m_{1}$ and $m_{2}$ is equal, then $m_{2}-m_{1}$ is
even, and if both $\ell_{1}$ and $\ell_{2}$ either odd or even then $h$
should be even. If the parity of $\ell_{1}$ and $\ell_{2}$ is equal and at
the same time the parity of $m_{1}$ and $m_{2}$ is also equal, then under (%
\ref{Equa_al}) we have 
\begin{align*}
C\left( \ell_{1},m_{1};\ell_{2},m_{2}\right) & =\mathsf{E}%
Y_{\ell_{1}}^{m_{1}}\left( \underline{\widetilde{X}}\right)
Y_{\ell_{2}}^{m_{2}}\left( \underline{\widetilde{X}}\right) ^{\ast} \\
& =\left( -1\right) ^{m_{2}}\sum_{\left\vert \ell_{1}-\ell_{2}\right\vert
/2\leq h\leq\left( \ell_{1}+\ell_{2}\right) /2}\sqrt{\frac{\left( 2\ell
_{1}+1\right) \left( 2\ell_{2}+1\right) }{4\pi\left( 4h+1\right) }}%
C_{\ell_{1},0;\ell_{2},0}^{2h,0}C_{\ell_{1},m_{1};%
\ell_{2},-m_{2}}^{2h,m_{1}-m_{2}}a_{2h}^{m_{1}-m_{2}} \\
& =\frac{\delta_{\ \ell_{1}=\ell_{2}}\delta_{m_{1}=m_{2}}}{4\pi}.
\end{align*}
On the other hand, if the parity of $\ell_{1}$ and $\ell_{2}$ are different,
then $h$ should be odd and $a_{h}^{m_{1}-m_{2}}\neq0$, if $m_{2}-m_{1}$ is
odd, i.e. the parity of $m_{1}$ and $m_{2}$ are also different, in which
case 
\begin{align*}
C\left( \ell_{1},m_{1};\ell_{2},m_{2}\right) & =\mathsf{E}%
Y_{\ell_{1}}^{m_{1}}\left( \underline{\widetilde{X}}\right)
Y_{\ell_{2}}^{m_{2}}\left( \underline{\widetilde{X}}\right) ^{\ast} \\
& =\left( -1\right) ^{m_{2}}\sum_{\left( \left\vert \ell_{1}-\ell
_{2}\right\vert -1\right) /2\leq h\leq\left( \ell_{1}+\ell_{2}-1\right) /2}%
\sqrt{\frac{\left( 2\ell_{1}+1\right) \left( 2\ell_{2}+1\right) }{4\pi\left(
4h+3\right) }}C_{\ell_{1},0;\ell_{2},0}^{2h\boldsymbol{+1}%
,0}C_{\ell_{1},m_{1};\ell_{2},-m_{2}}^{2h\boldsymbol{+1}%
,m_{1}-m_{2}}a_{2h+1}^{m_{1}-m_{2}} \\
& =0.
\end{align*}
Define $\widehat{\underline{A}_{L,1}}\left( n\right) =\left[ \widehat{%
\underline{a}_{1}}\left( n\right) ,\widehat{\underline{a}_{3}}\left(
n\right) ,\ldots,\widehat{\underline{a}_{2L+1}}\left( n\right) \right] $,
where $\widehat{\underline{a}_{2\ell+1}}\left( n\right) =\left[ \widehat{%
a_{2\ell+1}^{-2\ell}},\widehat{a_{2\ell+1}^{-2\ell+2}},\ldots ,\widehat{%
a_{2\ell+1}^{2\ell}}\right] $, and similarly $\widehat{\underline{A}_{L,2}}%
\left( n\right) =\left[ \widehat{\underline{a}_{2}}\left( n\right) ,\widehat{%
\underline{a}_{4}}\left( n\right) ,\ldots,\widehat{\underline{a}_{2L}}\left(
n\right) \right] $, where $\widehat{\underline{a}_{2\ell}}\left( n\right) =%
\left[ \widehat{a_{2\ell}^{-2\ell+1}},\widehat{a_{2\ell}^{-2\ell+3}},\ldots ,%
\widehat{a_{2\ell}^{2\ell-1}}\right] $. Let $\widehat{\underline{A}_{L}}%
\left( n\right) =\left[ \widehat{\underline{A}_{L,1}}\left( n\right) ,%
\widehat{\underline{A}_{L,2}}\left( n\right) \right] ^{\top}$ with dimension 
$\left( L+1\right) \left( 2L+1\right) $. Notice that the $\mathcal{C=}%
\limfunc{Var}\left( \widehat{\underline{A}_{L,1}}\left( n\right) \right) $
matrix is diagonal and we conclude

\begin{lemma}
If $\underline{\widetilde{X}}$ has density $f$ which is symmetric with
respect to equatorial plane then $2n\widehat{\underline{A}_{L}}\left(
n\right) ^{\ast}\mathcal{C}^{-1}\widehat{\underline{A}_{L}}\left( n\right) $
is $\chi_{\left( L+1\right) \left( 2L+1\right) }^{2}$ distributed with $%
\left( L+1\right) \left( 2L+1\right) $ degrees of freedom.
\end{lemma}

We reject this hypothesis of symmetry around the equatorial plane if $2n%
\widehat{\underline{A}_{L}}\left( n\right) ^{\ast}\mathcal{C}^{-1}\widehat{%
\underline{A}_{L}}\left( n\right) $ is large enough.

\subsection{Reflection with respect to meridial plane $\protect\varphi=%
\protect\varphi_{0}$ and $\protect\varphi=\protect\pi+\protect\varphi_{0}$}

\begin{definition}
$f$ is symmetric with respect to meridial plane $\varphi=\varphi_{0}$, and $%
\varphi=\pi+\varphi_{0}$, if for any $\underline{\widetilde{x}},\underline{%
\widetilde{x}}^{\prime}\in\mathbb{S}_{2}$, such that $\underline{\widetilde{x%
}}\left( \vartheta,\varphi\right) \rightarrow \underline{\widetilde{x}}%
^{\prime}\left( \vartheta,2\varphi_{0}-\varphi\right) $, then $f\left( 
\underline{\widetilde{x}}\right) =f\left( \underline{\widetilde{x}}%
^{\prime}\right) $.
\end{definition}

If $f$ is rotationally symmetric then it is symmetric with respect to
meridial planes which contains the axis $\underline{\widetilde{x}}_{0}$ of
symmetry.

\begin{lemma}
\label{Lemma_Merid_plane}$f$ is symmetric with respect to meridial plane $%
\varphi=\varphi_{0}$, and $\varphi=\pi+\varphi_{0}$, if and only if $a_{\ell
}^{m}=\left\vert a_{\ell}^{m}\right\vert e^{-im\left(
\varphi_{0}+\delta_{m>0}\pi\right) }$ for all $\ell$, and $m$. In this case 
\begin{equation*}
f\left( \underline{\widetilde{x}}\right)
=\sum_{\ell=0}^{\infty}\sum_{m=-\ell}^{\ell}\left\vert
a_{\ell}^{m}\right\vert e^{-im\left( \varphi_{0}+\delta_{m>0}\pi\right)
}Y_{\ell}^{m}\left( \underline{\widetilde{x}}\right) .
\end{equation*}
\end{lemma}

See Appendix \ref{Appendix_Proofs3} for the proof.

\begin{remark}
If $\varphi_{0}=\pi$, then all coefficients $a_{\ell}^{m}$ are real.
\end{remark}

\begin{remark}
Mean direction $\underline{\widetilde{\mu}}$ for a density which is
symmetric with respect to meridial plane $\varphi=\varphi_{0}$, and $%
\varphi=\pi +\varphi_{0}$, 
\begin{align*}
\mathsf{E}\widetilde{\underline{X}} & =\sqrt{2\pi/3}\left(
a_{1}^{1}-a_{1}^{-1},-i\left( a_{1}^{1}+a_{1}^{-1}\right) ,\sqrt{2}%
a_{1,0}\right) \\
& =\sqrt{2\pi/3}\left( \left\vert a_{1}^{1}\right\vert \left( e^{-i\left(
\varphi_{0}+\pi\right) }-e^{i\varphi_{0}}\right) ,-i\left\vert
a_{1}^{1}\right\vert \left( e^{-i\left( \varphi_{0}+\pi\right) }+e^{i\varphi
_{0}}\right) ,\sqrt{2}a_{1,0}\right) \\
& =\sqrt{2\pi/3}\left( -2\left\vert a_{1}^{1}\right\vert \cos\varphi
_{0},2\left\vert a_{1}^{1}\right\vert \sin\varphi_{0},\sqrt{2}a_{1,0}\right)
,
\end{align*}
belongs to the meridial plane $\varphi=\varphi_{0}$, and $\varphi=\pi
+\varphi_{0}$, which is very reasonable. Mean direction characterizes the
plane, i.e. the angle $\varphi_{0}$. If $\varphi_{0}=\pi$, then $\mathsf{E}%
\widetilde{\underline{X}}=\sqrt{2\pi/3}\left( -2\left\vert
a_{1}^{1}\right\vert ,0,\sqrt{2}a_{1,0}\right) ^{\top}.$
\end{remark}

\subsubsection{Testing reflection with respect to meridial plane $\protect%
\varphi=\protect\varphi_{0}$, and $\protect\varphi=\protect\pi+\protect%
\varphi_{0}$}

In this case the imaginary part of $\mathsf{E}e^{im\left(
\varphi_{0}+\delta_{m>0}\pi\right) }\widehat{a_{\ell}^{m}}=\rho_{\ell}^{m}$
is zero, i.e. 
\begin{equation*}
\func{Im}\mathsf{E}e^{im\left( \varphi_{0}+\delta_{m>0}\pi\right) }\widehat{%
a_{\ell}^{m}}=0.
\end{equation*}
Based on the sample $\underline{\widetilde{x}}_{1}\left( \vartheta
_{1},\varphi_{1}\right) $, $\underline{\widetilde{x}}_{2}\left(
\vartheta_{2},\varphi_{2}\right) ,\ldots,\underline{\widetilde{x}}_{n}\left(
\vartheta_{n},\varphi_{n}\right) $, we consider the statistics%
\begin{equation*}
2\func{Im}e^{im\left( \varphi_{0}+\pi\right) }\widehat{a_{\ell}^{m}}=\left(
-1\right) ^{m}e^{im\varphi_{0}}\widehat{a_{\ell}^{m}}+e^{-im\varphi_{0}}%
\widehat{a_{\ell}^{-m}},
\end{equation*}
and reject the hypothesis of symmetry with respect to meridial plane if
these are significantly large.

\section{Real Data Example -- Sunspot Activity}

The solar photospheric activity is a long-standing subject of observations
and research in Astronomy. We consider the data on Sunspots which contains
daily positions and areas of sunspots. This data can be found in the
Debrecen Photoheliographic Data (DPD) sunspot catalogue \newline
(http://fenyi.solarobs.csfk.mta.hu/DPD/, \cite{baranyi2016line}, \cite%
{gyHori2016comparative}), \cite{gyenge2014variations} .

Locations of a spot refer to the position of the centroid of the whole spot
or that of the umbra if an umbra is identified within a spot. Locations are
defined by their Heliographic latitude and Heliographic longitude. We use
daily data labelled \textquotedblleft sDPD" and consider four columns:
column No. 8 with NOAA sunspot group number, column No 9 with spot numbers
within the group, column No.14 with Heliographic latitude which is positive:
North, negative: South, and finally column No.15 with Heliographic
longitude. The daily data contains the same spot as many times as its
lifetime in days. We transformed the data such that each location is
included only once. This way the data set between 1976-2014 includes 187,223
positions, see Figure \ref{Sunspot3}.

\FRAME{ftbphFU}{6.6893in}{1.7694in}{0pt}{\Qcb{From left to right: Sunsopt
data, the histogram and the fitted density}}{\Qlb{Sunspot3}}{sunspot3.tif}{%
\special{language "Scientific Word";type "GRAPHIC";maintain-aspect-ratio
TRUE;display "USEDEF";valid_file "F";width 6.6893in;height 1.7694in;depth
0pt;original-width 13.5101in;original-height 3.5414in;cropleft "0";croptop
"1";cropright "1";cropbottom "0";filename 'SunSpot3.tif';file-properties
"XNPEU";}}

The histogram of size $3072$ (resolution number: $16$) shows two girdles
with equal distances from the equator, see Figure \ref{Sunspot3}. From this
figure, we might surmise that a generalization of the Dimroth-Watson
Distribution of the following form, would provide a good fit. 
\begin{equation}
f\left( \underline{\widetilde{x}};\gamma,\alpha\right) \cong\frac{1}{2}%
e^{\gamma\cos^{2}\left( \vartheta-\alpha\right) }+\frac{1}{2}%
e^{\gamma\cos^{2}\left( \vartheta+\alpha\right) },  \label{MixedDW}
\end{equation}
where $\alpha\in\left[ 0,\pi/2\right] $ and $\gamma<0$. Here $\alpha$ moves
the girdle up and down and $\left\vert \gamma\right\vert $ is the parameter
of concentration. The model (\ref{MixedDW}) can be considered as a
particular case of a more general mixture of Dimroth-Watson Distributions,
viz.%
\begin{equation*}
f_{W}\left( \underline{\widetilde{x}};\gamma,\alpha\right) \cong
pe^{\gamma_{1}\cos^{2}\left( \vartheta-\alpha_{1}\right) }+\left( 1-p\right)
e^{\gamma_{2}\cos^{2}\left( \vartheta+\alpha_{2}\right) },
\end{equation*}
where $p\in\left[ 0,1\right] $, $\alpha_{1},\alpha_{2}\in\left[ 0,\pi/2%
\right] $ and $\gamma_{1}$, $\gamma_{2}<0$. The simulation of such a model
is quite straightforward, starting with the simulation of two DW random
variates with parameters $\gamma_{1},\gamma_{2}$, then shift them by $%
\alpha_{1},\alpha_{2}$ respectively and finally mix them up by $p$ and $1-p $%
.

We used the histogram $H\left( \widetilde{\underline{x}}\right) $, see \cite%
{TGy_RSJ_WB2017}, for estimation of the parameter $\alpha\ $of the model (%
\ref{MixedDW}). First we took the average of $H\left( \widetilde{\underline{x%
}}\right) $ by longitudes for each fixed colatitude, then the maximum value
gives an estimate $\widehat{\alpha}$, which comes out to be $\widehat{\alpha}%
=0.2527$.

Now, we estimate $\gamma$ for a given $\alpha$. The model (\ref{MixedDW}) is
rotationally symmetric with axis $\underline{\widetilde{N}}$ therefore the
series expansion of the density has the form 
\begin{equation*}
f\left( \underline{\widetilde{x}};\gamma,\alpha\right) =\sum_{\ell
=0}^{\infty}c_{2\ell}\left( \gamma,\alpha\right) \frac{2\ell+1}{4\pi }%
P_{2\ell}\left( \cos\left( \vartheta\right) \right) .
\end{equation*}
We estimate the coefficients $c_{2\ell}\left( \gamma,\alpha\right) $ from
the data for $\ell=1,2,\ldots,10$, and applied the method of nonlinear least
squares for fitting (\ref{MixedDW}), obtaining $\widehat{\gamma}=-39.0022$.

\FRAME{ftbphFU}{5.6948in}{2.0963in}{0pt}{\Qcb{The yearly estimated shift $%
\protect\alpha$, $\protect\widehat{\protect\alpha}\left( t\right) $ upper
curve and $-\protect\widehat{\protect\alpha}\left( t\right) $ lower curve, $%
t=1976,1977,\ldots ,2014$}}{\Qlb{Butterfly}}{butterflyplotalpha3.eps}{%
\special{ language "Scientific Word"; type "GRAPHIC"; maintain-aspect-ratio
TRUE; display "USEDEF"; valid_file "F"; width 5.6948in; height 2.0963in;
depth 0pt; original-width 7.6562in; original-height 2.802in; cropleft "0";
croptop "1"; cropright "1"; cropbottom "0"; filename
'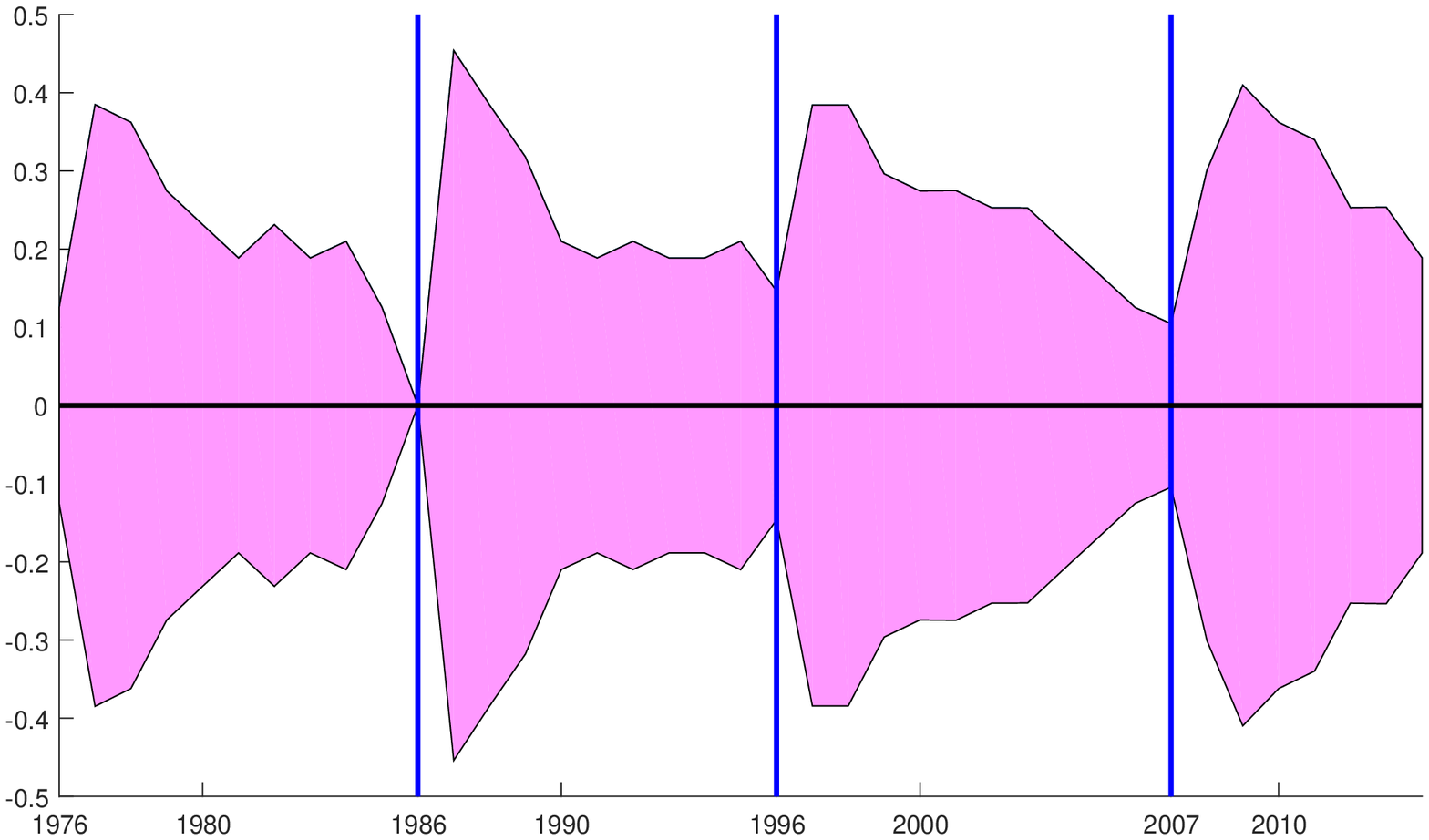';file-properties "XNPEU";}}

\begin{remark}
It is known that the Sunspot activity has a period of around 11 years, while
a more precise period can be obtained using Carrington rotation numbers, 
\cite{gyenge2014variations}. We estimated the shift $\alpha $ yearly between
1976-2014 years and plotted both $\widehat{\alpha }$ and $-\widehat{\alpha }$%
. The period of roughly 11 years shows up in the "Butterfly" Figure \ref%
{Butterfly}, which refers to the movement of girds by the years.
\end{remark}

\begin{acknowledgement}
This work was supported in part by the project EFOP-3.6.2-16-2017-00015
\end{acknowledgement}

\appendix

\subsection{Legendre polynomials and Spherical Harmonics}

\subsubsection{Legendre polynomials\label{LegendrePoly}}

Here we give a brief description and some properties of \textbf{Standardized
Legendre polynomials } (see e.g. \cite{Prudnikov30}, 2.17.5.1 for further
details)%
\begin{equation}
P_{\ell}\left( x\right) =\frac{1}{2^{\ell}\ell!}\frac{d^{\ell}\left(
x^{2}-1\right) ^{\ell}}{dx^{\ell}},\;x\in\left[ -1,1\right] ,
\label{Legendre_poly}
\end{equation}

$P_{0}\left( x\right) =1$, $P_{\ell}\left( 1\right) =1$, $P_{1}\left(
x\right) =x$, $P_{2}\left( x\right) =\left( 3x^{2}-1\right) /2$, (\cite%
{Erdelyi11a}, 2, p.180) they are orthogonal and 
\begin{equation*}
\int_{-1}^{1}\left[ P_{\ell}\left( x\right) \right] ^{2}dx=\frac{2}{2\ell+1},
\end{equation*}
hence%
\begin{equation}
\int_{\mathbb{S}_{2}}\left[ P_{\ell}\left( \cos\vartheta\right) \right]
^{2}\Omega\left( d\underline{\widetilde{x}}\right) =\frac{4\pi}{2\ell+1}.
\label{LegengreNormSqu}
\end{equation}

\subsubsection{More on Spherical Harmonics-- the Funk-Hecke formula and the
Clebsch-Gordan series}

At the \textit{North pole} $\underline{\widetilde{N}}$ when $\vartheta=0$
and $\varphi$ arbitrary, all spherical harmonics have the special value 
\begin{equation}
Y_{\ell}^{m}\left( \underline{\widetilde{N}}\right) =\delta_{m=0}\sqrt {%
\frac{2\ell+1}{4\pi}}.  \label{Y_North}
\end{equation}
Also from Equation (\ref{SpheriHarmOrthoNorm}), it follows%
\begin{align}
Y_{\ell}^{m\ast}\left( \vartheta,\varphi\right) & =Y_{\ell}^{m}\left(
\vartheta,-\varphi\right)  \label{Y_conjugate} \\
& =\left( -1\right) ^{m}Y_{\ell}^{-m}\left( \vartheta,\varphi\right) , 
\notag \\
Y_{\ell}^{-m}\left( \vartheta,\varphi\right) & =\left( -1\right)
^{m}e^{-i2m\varphi}Y_{\ell}^{m}\left( \vartheta,\varphi\right) .  \notag
\end{align}
For any two points $\underline{\widetilde{x}}_{1}\ $and $\underline{%
\widetilde{x}}_{2}$ from $\mathbb{S}_{2}$, we have the \textbf{addition
formula} (see \cite{Varshalovich1988}) 
\begin{equation}
\sum_{m=-\ell}^{\ell}Y_{\ell}^{m\ast}\left( \underline{\widetilde{x}}%
_{1}\right) Y_{\ell}^{m}\left( \underline{\widetilde{x}}_{2}\right) =\frac{%
2\ell+1}{4\pi}P_{\ell}\left( \cos\gamma\right) ,  \label{FormulaAddition}
\end{equation}
where $\cos\gamma=\underline{\widetilde{x}}_{1}\cdot\underline{\widetilde{x}}%
_{2}$. In particular 
\begin{equation}
\sum_{m=-\ell}^{\ell}Y_{\ell}^{m\ast}\left( \underline{\widetilde{x}}\right)
Y_{\ell}^{m}\left( \underline{\widetilde{x}}\right) =\frac{2\ell+1}{4\pi}.
\label{Ylm_Sum_L}
\end{equation}

Suppose $G$ is continuous on $\left[ -1,1\right] $, then $G\left(
\cos\gamma\right) =G\left( \underline{\widetilde{x}}_{1}\cdot \underline{%
\widetilde{x}}\right) $ is defined on $\mathbb{S}_{2}$, where $\underline{%
\widetilde{x}}_{1}$ is fixed and $\underline{\widetilde{x}}_{1}\cdot%
\underline{\widetilde{x}}=\cos\gamma$. The series expansion in terms of
Legendre polynomials, 
\begin{equation*}
G\left( \cos\gamma\right) =\sum_{\ell=0}^{\infty}\frac{2\ell+1}{2}%
\int_{-1}^{1}G\left( x\right) P_{\ell}\left( x\right) dxP_{\ell}\left(
\cos\gamma\right) ,
\end{equation*}
can be derived with the help of \textbf{Funk-Hecke formula} which says (\cite%
{Mueller1966}, p. 20)%
\begin{align}
\int_{\mathbb{S}_{2}}G\left( \underline{\widetilde{x}}_{1}\cdot \underline{%
\widetilde{x}}\right) Y_{\ell}^{m}\left( \underline{\widetilde{x}}\right)
\Omega\left( d\underline{\widetilde{x}}\right) & =G_{\ell
}Y_{\ell}^{m}\left( \underline{\widetilde{x}}_{1}\right) ,
\label{Formula_Funk-Hecke} \\
G_{\ell} & =2\pi\int_{-1}^{1}G\left( x\right) P_{\ell}\left( x\right) dx. 
\notag
\end{align}

\smallskip An important result that has been frequently used in this paper
and which provides the coefficients of the product of two spherical
harmonics $Y_{\ell_{1}}^{m_{1}}\left( \underline{\widetilde{x}}\right)
Y_{\ell_{2}}^{m_{2}}\left( \underline{\widetilde{x}}\right) ^{\ast}$ in
terms of linear combination of other spherical harmonics, is the so-called 
\textbf{Clebsch-Gordan series} (see \cite{Varshalovich1988} p.144), and is
given by 
\begin{equation}
Y_{\ell_{1}}^{m_{1}}\left( \underline{\widetilde{x}}\right)
Y_{\ell_{2}}^{m_{2}}\left( \underline{\widetilde{x}}\right) ^{\ast}=\left(
-1\right) ^{m_{2}}\sum_{\left\vert \ell_{1}-\ell_{2}\right\vert \leq
k\leq\ell_{1}+\ell_{2}}\sqrt{\frac{\left( 2\ell_{1}+1\right) \left(
2\ell_{2}+1\right) }{4\pi\left( 2k+1\right) }}C_{\ell_{1},0;%
\ell_{2},0}^{k,0}C_{\ell_{1},m_{1};%
\ell_{2},-m_{2}}^{k,m_{1}-m_{2}}Y_{k}^{m_{1}-m_{2}}\left( \underline{%
\widetilde{x}}\right) .  \label{Clebsch_Gordan_series}
\end{equation}
These quantities $C_{\ell_{1},m_{1};\ell_{2},m_{2}}^{\ell,m}$ are called the
Clebsch-Gordan coefficients and can be evaluated sometimes using available
MATLAB codes. For some applications involving these for 3D spectra on
sphere, see \cite{Marinucci2010}, \cite{Ter2015},.

Some basic properties as follows 
\begin{equation*}
\sum_{m_{1:2}=-\ell_{1:2}}^{\ell_{1:2}}C_{\ell_{1},m_{1};\ell_{2},m_{2}}^{%
\ell,m}C_{\ell_{1},m_{1};\ell_{2},m_{2}}^{\ell^{\ast},m^{\ast}}=\delta
_{\ell=\ell^{\ast}}\delta_{m=m^{\ast},}
\end{equation*}
\cite{Varshalovich1988} p.250.$C_{\ell_{1},0;\ell_{2},0}^{0,0}=\delta
_{\ell_{1}=\ell_{2}}\left( -1\right) ^{\ell_{1}}/\sqrt{2\ell_{1}+1}$, and
8.7.2.4 p.259%
\begin{equation}
C_{\ell_{1},m_{1};\ell_{2},m_{2}}^{0,0}=\delta_{m_{1}=-m_{2}}\delta_{\ell
_{1}=\ell_{2}}\frac{\left( -1\right) ^{\ell_{1}-m_{1}}}{\sqrt{2\ell_{1}+1}}.
\label{Cl_Gord_00}
\end{equation}
\cite{Varshalovich1988} 8.5.1 (a) p.248.

\textbf{Wigner D-matrices} $D_{m,k}^{\left( \ell\right) }\left( g\right) $,
see \cite{Varshalovich1988}, p. 79 for details. We introduce the notation $%
D^{\left( \ell\right) }=\left[ D_{m,k}^{\left( \ell\right) }\right] $, for
fixed rotation $g$. Thus $D^{\left( \ell\right) }$ denotes a unitary matrix
of order $2\ell+1$, and it follows $D^{\left( \ell\right) }\left[ D^{\left(
\ell\right) }\right] ^{-1}=D^{\left( \ell\right) }D^{\left( \ell\right)
\ast} $, $\det D^{\left( \ell\right) }=1$ (unimodular). We shall use the
integral 
\begin{equation}
\int_{SO\left( 3\right) }D_{m,k}^{\left( \ell\right) }\left( g\right)
dg=\delta_{\ell,0}\delta_{m,0}\delta_{k,0},  \label{Int_D}
\end{equation}
where $dg=\sin\vartheta d\vartheta d\varphi d\gamma/8\pi^{2}$ is the Haar
measure (see \cite{SteinWeiss} I.4.14).

\subsection{Some proofs}

\subsubsection{Example \protect\ref{Ex_vMF} \label{Appendix_Proofs_vMF}}

\begin{proof}
The Funk-Hecke formula (see (\ref{Formula_Funk-Hecke})) gives us 
\begin{equation*}
\int_{\mathbb{S}_{2}}f\left( \underline{\widetilde{x}}\cdot \underline{%
\widetilde{\mu}};\kappa\right) Y_{\ell}^{m\ast}\left( \underline{\widetilde{x%
}}\right) \Omega\left( d\underline{\widetilde{x}}\right)
=c_{\ell}Y_{\ell}^{m\ast}\left( \underline{\widetilde{\mu}}\right) ,
\end{equation*}
where 
\begin{equation*}
c_{\ell}=2\pi\int_{-1}^{1}f\left( y;\kappa\right) P_{\ell}\left( y\right) dy.
\end{equation*}
Now%
\begin{equation*}
\int_{-1}^{1}\exp\left( \kappa y\right) P_{\ell}\left( y\right) dy=\sqrt{%
\frac{2\pi}{\kappa}}I_{\ell+1/2}\left( \kappa\right) ,
\end{equation*}
hence $c_{\ell}=I_{\ell+1/2}\left( \kappa\right) /I_{1/2}\left(
\kappa\right) $. Plugging in $c_{\ell}Y_{\ell}^{m\ast}$ into the series
expansion (\ref{Series_f_L}) for $f$, we have 
\begin{equation*}
f\left( \underline{\widetilde{x}};\underline{\widetilde{\mu}},\kappa\right)
=\sum_{\ell=0}^{\infty}\sum_{m=-\ell}^{\ell}c_{\ell}Y_{\ell}^{m\ast}\left( 
\underline{\widetilde{\mu}}\right) Y_{\ell}^{m}\left( \underline{\widetilde{x%
}}\right) .
\end{equation*}
Finally applying the addition formula for the spherical harmonics (see (\ref%
{FormulaAddition})) we obtain, 
\begin{align*}
f\left( \underline{\widetilde{x}};\underline{\widetilde{\mu}},\kappa\right)
& =\sum_{\ell=0}^{\infty}c_{\ell}\frac{2\ell+1}{4\pi}P_{\ell}\left( 
\underline{\widetilde{x}}\cdot\underline{\widetilde{\mu}}\right) \\
& =\sum_{\ell=0}^{\infty}\frac{2\ell+1}{4\pi}\frac{I_{\ell+1/2}\left(
\kappa\right) }{I_{1/2}\left( \kappa\right) }P_{\ell}\left( \cos
\gamma\right) ,
\end{align*}
and using (\ref{Y2P}) gives us the desired result (\ref{Density_Fisher}).
\end{proof}

\subsubsection{Example \protect\ref{EX_Ylm_real_suqare} \label%
{Appendix_Proofs1}}

\begin{proof}
Note that for $m>0$%
\begin{align*}
2Y_{\ell,m}^{2} & =\left( Y_{\ell}^{m}+\left( -1\right)
^{m}Y_{\ell}^{-m}\right) ^{2} \\
& =\left( Y_{\ell}^{m}\right) ^{2}+\left( Y_{\ell}^{-m}\right)
^{2}+2\left\vert Y_{\ell}^{m}\right\vert ^{2},
\end{align*}
and by Clebsch-Gordan series (\ref{Clebsch_Gordan_series}) we have 
\begin{align}
\left( Y_{\ell}^{m}\right) ^{2} & =Y_{\ell}^{m}\left( \left( -1\right)
^{m}Y_{\ell}^{-m}\right) ^{\ast}  \notag \\
& =\frac{2\ell+1}{\sqrt{4\pi}}\sum_{h=m}^{\ell}\sqrt{\frac{1}{4h+1}}%
C_{\ell,0;\ell,0}^{2h,0}C_{\ell,m;\ell,m}^{2h,2m}Y_{2h}^{2m},
\label{Ylm_suqare}
\end{align}
see \cite{Varshalovich1988}, 8.5, (h), p.250. Similarly $\left(
Y_{\ell}^{-m}\right) ^{2}=\left( \left( Y_{\ell}^{m}\right) ^{2}\right)
^{\ast}$ 
\begin{equation}
\left( Y_{\ell}^{-m}\right) ^{2}=\frac{2\ell+1}{\sqrt{4\pi}}\sum_{h=m}^{\ell}%
\sqrt{\frac{1}{4h+1}}C_{\ell,0;\ell,0}^{2h,0}C_{\ell,m;%
\ell,m}^{2h,2m}Y_{2h}^{-2m},  \label{Yl-m_suqare}
\end{equation}
which follows from $\left( Y_{\ell}^{m}\right) ^{2}$. Therefore using (\ref%
{Ylm_abs_suqare}), we have 
\begin{align}
Y_{\ell,m}^{2} & =\frac{2\ell+1}{2\sqrt{4\pi}}\sum_{h=m}^{\ell}\sqrt {\frac{1%
}{4h+1}}C_{\ell,0;\ell,0}^{2h,0}C_{\ell,m;\ell,m}^{2h,2m}\left(
Y_{2h}^{2m}+Y_{2h}^{-2m}\right)  \notag \\
& +\left( -1\right) ^{m}\frac{2\ell+1}{\sqrt{4\pi}}\sum_{h=0}^{\ell}\sqrt{%
\frac{1}{4h+1}}C_{\ell,0;\ell,0}^{2h,0}C_{\ell,m;\ell,-m}^{2h,0}Y_{2h}^{0}, 
\notag
\end{align}
and the required representation (\ref{Ylm_real_suqare}) follows.
\end{proof}

\subsubsection{Lemma \protect\ref{Lemma_Equa_plane}\label{Appendix_Proofs2}}

\begin{proof}
Consider the transformation $\underline{\widetilde{x}}\left( \vartheta
,\varphi\right) \rightarrow\underline{\widetilde{x}}^{\prime}\left(
\pi-\vartheta,\varphi\right) $, we have 
\begin{equation*}
Y_{\ell}^{m}\left( \pi-\vartheta,\varphi\right) =\left( -1\right)
^{\ell+m}Y_{\ell}^{m}\left( \vartheta,\varphi\right) .
\end{equation*}
Now using symmetry, we get 
\begin{align*}
f\left( \underline{\widetilde{x}}\right) &
=\sum_{\ell=0}^{\infty}\sum_{m=-\ell}^{\ell}a_{\ell}^{m}Y_{\ell}^{m}\left( 
\underline{\widetilde{x}}\right) \\
& =f\left( \underline{\widetilde{x}}^{\prime}\right) \\
& =\sum_{\ell=0}^{\infty}\sum_{m=-\ell}^{\ell}\left( -1\right) ^{\ell
+m}a_{\ell}^{m}Y_{\ell}^{m}\left( \underline{\widetilde{x}}\right) .
\end{align*}
Using the fact that $a_{\ell}^{m}=\left( -1\right) ^{\ell+m}a_{\ell}^{m}$, $%
a_{2\ell+1}^{2m}=0$, and $a_{2\ell}^{2m+1}=0$, we get 
\begin{equation*}
f\left( \underline{\widetilde{x}}\right) =\sum_{\ell=0}^{\infty}\left(
\sum_{m=-\ell}^{\ell}a_{2\ell}^{2m}Y_{2\ell}^{2m}\left( \underline{%
\widetilde{x}}\right)
+\sum_{m=-\ell-1}^{\ell}a_{2\ell+1}^{2m+1}Y_{2\ell+1}^{2m+1}\left( 
\underline{\widetilde{x}}\right) \right) .
\end{equation*}
\end{proof}

\subsubsection{Lemma \protect\ref{Lemma_Merid_plane}\label{Appendix_Proofs3}}

\begin{proof}
Consider the transformation $\underline{\widetilde{x}}\left( \vartheta
,\varphi\right) \rightarrow\underline{\widetilde{x}}^{\prime}\left(
\vartheta,2\varphi_{0}-\varphi\right) $, we have 
\begin{equation*}
Y_{\ell}^{m}\left( \vartheta,2\varphi_{0}-\varphi\right) =\left( -1\right)
^{m}e^{i2m\varphi_{0}}Y_{\ell}^{-m}\left( \vartheta,\varphi\right) .
\end{equation*}
Now under the assumption of the lemma 
\begin{align*}
f\left( \underline{\widetilde{x}}\right) &
=\sum_{\ell=0}^{\infty}\sum_{m=-\ell}^{\ell}a_{\ell}^{m}Y_{\ell}^{m}\left( 
\underline{\widetilde{x}}\right) \\
& =f\left( \underline{\widetilde{x}}^{\prime}\right) \\
& =\sum_{\ell=0}^{\infty}\sum_{m=-\ell}^{\ell}\left( -1\right)
^{m}e^{i2m\varphi_{0}}a_{\ell}^{m}Y_{\ell}^{-m}\left( \underline{\widetilde{x%
}}\right) .
\end{align*}
From the polar form $a_{\ell}^{m}=\rho_{\ell}^{m}e^{-i\varphi_{\ell}^{m}}$,
and from the equations $a_{\ell}^{m\ast}=\left( -1\right) ^{m}a_{\ell}^{-m}$%
, and $a_{\ell}^{-m}=\left( -1\right) ^{m}e^{i2m\varphi_{0}}a_{\ell}^{m}$,
we get 
\begin{align*}
-\varphi_{\ell}^{m} & =\varphi_{\ell}^{-m}+m\pi, \\
\varphi_{\ell}^{-m} & =\varphi_{\ell}^{m}+2m\varphi_{0}+m\pi.
\end{align*}
The solution of these equations is given by 
\begin{align*}
\varphi_{\ell}^{-m} & =m\varphi_{0}, \\
\varphi_{\ell}^{m} & =-m\left( \varphi_{0}+\pi\right)
\end{align*}
which can be combined into the result 
\begin{equation*}
\varphi_{\ell}^{m}=-m\left( \varphi_{0}+\delta_{m>0}\pi\right) ,
\end{equation*}
where $\delta_{m>0}$ denotes the Kronecker delta. Hence $a_{\ell}^{m}=\rho_{%
\ell}^{m}e^{-im\left( \varphi_{0}+\delta_{m>0}\pi\right) }$ .
\end{proof}

%

\bibliographystyle{aabbrv}
\bibliography{01Bibl_MM16}

\end{document}